\documentclass[preprint,3p,number]{elsarticle}

\usepackage{FACS-SCP}

\begin{document}

\begin{frontmatter}

\title{Unifying Semantic Foundations for \\ Automated Verification Tools in Isabelle/UTP}

\author{Simon Foster\corref{cor1}}
\ead{simon.foster@york.ac.uk}

\cortext[cor1]{Corresponding author}

\author{James Baxter\corref{}}
\ead{james.baxter@york.ac.uk}

\author{Ana Cavalcanti}
\ead{ana.cavalcanti@york.ac.uk}

\author{Jim Woodcock}
\ead{jim.woodcock@york.ac.uk}

\author{Frank Zeyda}

\address{Department of Computer Science, University of York, Deramore Lane, Heslington, York YO10 5GH, United Kingdom}

\begin{abstract}
  The growing complexity and diversity of models used for engineering dependable systems implies that a variety of
  formal methods, across differing abstractions, paradigms, and presentations, must be integrated. Such an integration
  requires unified semantic foundations for the various notations, and co-ordination of a variety of automated
  verification tools. The contribution of this paper is Isabelle/UTP, an implementation of Hoare and He's Unifying
  Theories of Programming, a framework for unification of formal semantics. Isabelle/UTP permits the mechanisation of
  computational theories for diverse paradigms, and their use in constructing formalised semantics. These can be
  further applied in the development of verification tools, harnessing Isabelle's proof automation facilities. Several
  layers of mathematical foundations are developed, including lenses to model variables and state spaces as algebraic
  objects, alphabetised predicates and relations to model programs, algebraic and axiomatic semantics, proof
  tools for Hoare logic and refinement calculus, and UTP theories to encode computational paradigms.
\end{abstract}

\begin{keyword}
Theorem Proving \sep Lenses \sep Unifying Theories of Programming \sep Hoare Logic \sep Isabelle/HOL
\end{keyword}

\end{frontmatter}

\section{Introduction}
\label{sec:intro}
Unifying Theories of Programming~\cite{Hoare&98} (UTP) is a framework for capturing, unifying, and integrating formal
semantics using predicate and relational calculus.  It aims to provide a ``coherent structure for computer
science''~\cite{Hoare&98}, by characterising the various programming languages it has produced, their foundational
computational paradigms, and different semantic flavours. Along one axis, UTP allows us to study computational
paradigms, such as functional, imperative, concurrent~\cite{Hoare&98,Oliveira&09}, real-time~\cite{Sherif2010}, and
hybrid dynamical systems~\cite{Foster16b,Foster17b,Foster19c-HybridRelations}.
In another axis, it allows us to characterise and link different presentations of semantics, such as axiomatic
semantics~\cite{Hoare69,Mor1996,Back1998,Reynolds2002}, with operational semantics and also algebraic semantics. UTP
thus unites a diverse range of notations and paradigms, and so, with adequate tool support, it allows us to answer the
substantial challenge of integrating formal
methods~\cite{Paige1997FM-IntegratedFormalMethods,Galloway1997-IntegratedFormalMethods,Broy1998-IntegratedFormalMethods}.

The contribution of this paper is the theoretical and practical foundations of a tool for building UTP-based
verification tools, called Isabelle/UTP, which draws from the strongest characteristics of previous
implementations~\cite{Oliveira07,Feliachi2010,Zeyda2012,Foster14,Zeyda16}. Our framework can unify the different
paradigms and semantic models needed for modelling heterogeneous systems, and provides facilities for constructing
verification tools. Isabelle/UTP is a shallow embedding~\cite{Gordon1989-VerifyHOL,vonWright1994-RefHOL} of UTP into
Isabelle/HOL~\cite{Isabelle}, and so has access to Isabelle's proof capabilities. Isabelle is highly extensible,
provides efficient automated proof~\cite{Blanchette2011}, and has facilities needed to develop plugins for performing
analysis on mathematical models of software and hardware.

Isabelle/UTP uses lenses~\cite{Foster07,Foster09} to algebraically characterise mutation of a program state using two
functions, similar to Back and von Wright's work~\cite{Back1998,Back2005-Procedures}, and to support semantic reasoning
facilities used in program verification and refinement. Lenses were originally developed to support bidirectional
programming languages for solving the ``view-update problem'' in database theory~\cite{Bancilhon1981-ViewUpdate}. Here,
we apply lenses to the modelling of state mutation and extend it with novel operators and laws.

Lenses allow us to generalise Back and von Wright's approach~\cite{Back1998} since we can abstractly model ``regions''
of the observable state. These regions may correspond to individual variables, sets of variables, and also
hierarchies. We develop several novel relations on lenses that allow us to relate these regions, including notions for
independence, containment, and equivalence. We also provide operators for composing regions in a set-theoretic manner,
which allows us to model alphabets of variables. Thus, lenses allow us to characterise syntax-related aspects of program
verification, such as free variables, substitution, frames, and aliasing. We mechanise an algebraic structure for lenses
in Isabelle/HOL, and show how it unifies a variety of state space models~\cite{Schirmer2009}. We draw comparisons with
Back and von Wright's variable axioms~\cite{Back1998,Back2005-Procedures} and separation
algebra~\cite{Calcagno2007}. Our account of lenses is a unifying algebra for observation and mutation of program state.

Upon our algebraic foundation of observation spaces, we construct UTP's relational program model. We develop a ``deep''
expression model~\cite{Zeyda16}, which is technically shallow, and yet has explicit syntax constructors supporting
inductive proofs. We develop UTP predicates and binary relations~\cite{Tarski71}, and provide a rich set of mechanically
verified algebraic theorems, including the famous ``laws of programming''~\cite{Hoare87}, which form the basis for
axiomatic semantics. Crucially, lenses allow us to express meta-logical provisos that involve syntax-related properties,
such as whether two variables are different or whether an expression depends on a particular variable, without needing a
deep embedding~\cite{Foster14,Zeyda16}.

We then develop several semantic presentations, including operational semantics, Hoare
logic~\cite{Hoare69,Gordon1989-VerifyHOL,Hoare&98}, and refinement calculus~\cite{Mor1996}, with linking theorems
showing how they are connected. From these, we also develop tactics for symbolic execution of relational
programs~\cite{GordonC2010}, verification using Hoare logic~\cite{Hoare&98}, and an example verification of a
find-and-replace algorithm. This illustrates the practicability and extensibility of our tool, which is not hampered by
our deep expression model and other meta-logical machinery.

Finally, we demonstrate the mechanisation of UTP theories within the relational model, which allows us to support a
hierarchy of advanced computational paradigms, including reactive programming~\cite{Foster17c}, hybrid
programming~\cite{Foster19c-HybridRelations,Foster2020-dL}, and object oriented programming~\cite{SCS06}.  Our work
significantly advances previous contributions~\cite{Zeyda2012} with UTP theories whose observation spaces are typed by
Isabelle, and which link to established libraries of algebraic theorems~\cite{HOL-Algebra,Armstrong2015} that facilitate
efficient derivation of programming laws.

This paper is an extension of two conference papers~\cite{Foster16a,Foster18b}.  We extend and refine the material on
lenses from~\cite{Foster16a}, including a precise account of the algebra, a substantial body of novel
theorems for each operator, a novel command for constructing lens-based state spaces, and additional motivating
examples. We extend~\cite{Foster18b} with additional theorems that can be derived for UTP theories, a complete example
based on timed relations with its mechanisation, and the use of frames in refinement calculus. On the whole, we simply
present the core theorems without proofs, and therefore refer the interested reader to a number of companion
reports~\cite{Foster18c-Optics,Foster19a-IsabelleUTP} and our repository\footnote{Isabelle/UTP Repository:
\url{https://github.com/isabelle-utp/utp-main}}. All the definitions, theorems, and proofs in this paper are
mechanically validated in Isabelle/HOL, and usually accompanied by an icon (\isalogo) linking to the corresponding
resources in our repository. Though our work is primarily based on Isabelle/HOL, we prefer to use more traditional
mathematical notations~\cite{Spivey89,Hoare&98} in this work\footnote{The precedence of operators in this paper, from
highest to lowest, is : $\usubapp$, $:=$, $\oast$, $\relsemi$, $\conditional{}{\cdot}{}$, $=$, $\neg$, $\land$, $\lor$,
$\implies$, $\iff$, $\unrest$.}, since we believe this makes the results more accessible.

In summary, our contributions are as follows: (1) mechanisation of lens theory in Isabelle/HOL, including fundamental
algebraic theorems, and extension with novel relations and combinators; (2) facilities for automating construction of
observation spaces; (3) an expression model, founded on lenses, providing both efficient proof and syntax-related
queries, such as free variables and substitution; (4) a generic relational program model and proven laws of programming;
(5) application to development of unified verification calculi, such as operational semantics, Hoare logic, and
refinement calculus, including treatment of aliasing and frames; (6) characterisation of mechanised UTP theories to
support various computational paradigms.

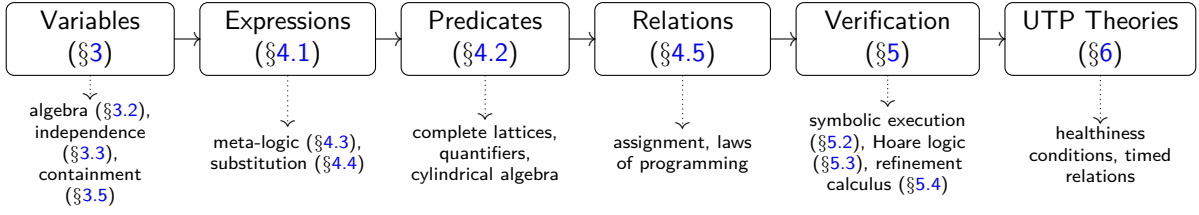
\begin{figure}
  \centering

  \begin{minipage}{.8\linewidth}
  \xymatrix@-1.2pc{
    *+[F-:<3pt>]{\parbox{2cm}{\centering \sf Variables (\S\ref{sec:lenses})}} \ar[r] \ar@{..>}[d] & 
    *+[F-:<3pt>]{\parbox{2.1cm}{\centering \sf Expressions (\S{\ref{sec:utp-expr}})}} \ar[r] \ar@{..>}[d] & 
    *+[F-:<3pt>]{\parbox{2cm}{\centering \sf Predicates (\S\ref{sec:predcalc})}} \ar[r] \ar@{..>}[d] & 
    *+[F-:<3pt>]{\parbox{2.1cm}{\centering \sf Relations (\S\ref{sec:relcalc})}} \ar[r] \ar@{..>}[d] & 
    *+[F-:<3pt>]{\parbox{2.2cm}{\centering \sf Verification (\S\ref{sec:verify-rel})}} \ar[r] \ar@{..>}[d] &
    *+[F-:<3pt>]{\parbox{2.3cm}{\centering \sf UTP Theories (\S\ref{sec:utp-thy})}} \ar@{..>}[d]
    \\
    \parbox{2cm}{\centering\scriptsize \sf algebra (\S\ref{sec:lensax}), independence (\S\ref{sec:lens-indep}), containment (\S\ref{sec:order-eq})} &
    \parbox{2.1cm}{\centering\scriptsize \sf
    meta-logic (\S\ref{sec:meta-logic}), substitution (\S\ref{sec:subst})} &
    \parbox{2cm}{\centering\scriptsize \sf complete lattices, quantifiers, cylindrical algebra} &
    \parbox{2.1cm}{\centering\scriptsize \sf assignment, laws of programming} &
    \parbox{2.2cm}{\centering\scriptsize \sf symbolic execution (\S\ref{sec:symexec}), Hoare logic (\S\ref{sec:hoarelogic}), refinement calculus (\S\ref{sec:ref-calc})} &
    \parbox{2.3cm}{\centering\scriptsize\sf healthiness conditions, timed relations}
  }
  \end{minipage}
  \caption{Roadmap of the Isabelle/UTP Foundations}
  \label{fig:roadmap}

  \vspace{-2ex}

\end{figure}

In \S\ref{sec:background} we motivate Isabelle/UTP, and review foundational work. In \S\ref{sec:lenses} to
\S\ref{sec:utp-thy}, we describe our contributions. The paper overview given in Figure~\ref{fig:roadmap} shows how the
foundational parts of Isabelle/UTP are connected, and the sections in which they are documented. In \S\ref{sec:lenses}
we describe how state spaces and variables in a program can be modelled algebraically using lenses.  In
\S\ref{sec:isabelle-utp} we describe the core of Isabelle/UTP, first defining its expression model in
\S\ref{sec:utp-expr}, predicates in \S\ref{sec:predcalc}, meta-logical facilities in \S\ref{sec:meta-logic} and
\S\ref{sec:subst}, and the relational program model in \S\ref{sec:relcalc} with proof support and mechanised laws of
programming.  In \S\ref{sec:verify-rel} we use this relational program model to build tools for symbolic evaluation,
verification using Hoare calculus with several examples in Isabelle/UTP, and also Morgan's refinement
calculus~\cite{Mor1996}.  In \S\ref{sec:utp-thy} we describe how different computational paradigms can be captured and
mechanised using Isabelle/UTP theories, illustrating this using a UTP theory for concurrent and reactive programs. In
\S\ref{sec:relwork} we survey related work, and in \S\ref{sec:concl}, we conclude.

\section{Preliminaries and Motivation}
\label{sec:background}
In this section we motivate the contributions in our paper and briefly survey foundational work. We first introduce UTP
(\cref{sec:UTP}) and Isabelle/HOL (\cref{sec:isahol}). In \cref{sec:verify-sem} we briefly survey work on semantic
embedding to support verification. This leads to the conclusion that the state space modelling approach is the main
consideration, and so in \cref{sec:state-space-modelling} we survey and critique different approaches. Finally in
\cref{sec:our-approach} we explain how Isabelle/UTP advances the state of the art.

\subsection{Unifying Theories of Programming}
\label{sec:UTP}

Several authors consider integration of formal methods, notably Broy~\cite{Broy1998-IntegratedFormalMethods} and
Paige~\cite{Paige1997FM-IntegratedFormalMethods}. These authors emphasise the centrality of unified formal
semantics~\cite{Hehner1988,Hehner1990}. Specifically, if diverse formal methods are to be coordinated, then their
various semantic models must be reconciled to ensure consistent
verification~\cite{Paige1997FM-IntegratedFormalMethods,Galloway1997-IntegratedFormalMethods}. An important development
is Hoare and He's Unifying Theories of Programming (UTP)~\cite{Hoare1994,Hoare&98,Cavalcanti&06} technique, which fuses
several intellectual streams, notably Hehner's predicative
programming~\cite{Hehner1984-PredicativeProgramming,Hehner1988,Hoare1984-ProgramsArePredicates}, relational
calculus~\cite{Tarski41}, and the refinement calculi of Back and von Wright~\cite{Back1998}, and Morgan~\cite{Mor1996}.

The goal of UTP is to find the fundamental computational paradigms that underlie programming and modelling languages and
characterise them with unifying denotational and algebraic semantics. A significant precursor to UTP is a seminal paper
entitled \textit{Laws of Programming}~\cite{Hoare87}, in which nine prominent computer scientists, led by Hoare, give a
complete set of algebraic laws for GCL~\cite{Dijkstra75}. The purpose of the laws of programming, however, is much
deeper: it is to find laws that unify different programming languages.

UTP uses binary relations to model programs as predicates~\cite{Hehner1984-PredicativeProgramming,Hehner1990}. Relations
map an initial value of a variable, such as $x$, to its later value, $x'$. These might model program variables or
observations of the real world. For example, $clock : \nat$ records the passage of time. It makes no sense to assign
values to $clock$ that reverse time. Healthiness conditions constrain observations using idempotent functions on
predicates. For example, application of $\healthy{HT}(P) \defs (P \land clock \le clock')$ gives a predicate that
forbids reverse time travel. If $P$ is unhealthy, for example $clock' = clock - 1$, then application of
$\healthy{HT}$ changes its meaning: $\healthy{HT}(clock = clock' + 1) = \false$.

The choice of observational variables and healthiness conditions defines a UTP theory. It characterises the set of
relations whose elements are the fixed points of the healthiness conditions; $\healthy{HT}(P) = P$ for our example. An
algebra gives the relationship between theory elements, which supports the laws of programming for a particular
paradigm~\cite{Hoare87}. UTP theories can be combined by composing their healthiness conditions, which allows
multi-paradigm semantics~\cite{Oliveira&09}.

Reactive processes~\cite[chapter~8]{Hoare&98} is a UTP theory for event-driven programs. It includes observational
variables $wait : \bool$ and $tr : \seq Event$, which represent, respectively, whether a program is quiescent
(i.e. waiting for interaction), and the sequence of observed events (drawn from $Event$). An example healthiness
condition is $\healthy{R1}(P) \defs tr \le tr'$, where $\le$ is the prefix order on sequences. $\healthy{R1}$ states
that the interaction history must be retained -- we may not undo past events.

A mechanisation of UTP in a theorem prover, like Isabelle~\cite{Isabelle,Feliachi2010,Foster14}, allows us to develop
verification tools from UTP theories. Here, there are two goals that must be balanced: (1) the ability to engineer UTP
theories and express the resulting algebraic laws; and (2) support for efficient automated
verification~\cite{Gordon1989-VerifyHOL,vonWright1994-RefHOL,Alkassar2008}. In addition, the UTP is not about one
programming language, or even one intermediate verification language (IVL)~\cite{Barnett2005Boogie,Filliatre2013-Why3},
but unification of diverse modelling and program paradigms through algebraic laws. For example, our mechanisation should
permit the algebraic law below~\cite[page~96]{Back1998}.

\begin{example}[Commutativity of Assignments] \label{ex:asn-comm}
\begin{align*}
  (x := e \relsemi y := f) &= (y := f \relsemi x := e) & \text{provided } x \text{ and } y \text{ are independent}, x \text{ is not free in } f, y \text{ is not free in } e.
\end{align*}
\end{example}

\noindent This law states that two assignments, $x := e$ and $y := f$, can commute provided that they are made to
distinct variables, $x$ and $y$, and they are not mentioned in $f$ and $e$, respectively. Now, such a law is intuitive
and holds in a variety of languages, which makes it a useful artefact for unification. Mechanising a law like this
requires that variables are modelled as first-class citizens, so that they are objects of the logic, with which we can
formulate the side conditions. In \cref{sec:verify-sem} we consider approaches to semantic embeddings, and motivate the
approach we have chosen, but first we consider Isabelle/HOL.

\subsection{Isabelle/HOL}
\label{sec:isahol}

Isabelle/UTP is a conservative extension of Isabelle/HOL~\cite{Isabelle}, which is a proof assistant for
Higher Order Logic (HOL). It consists of the Pure meta-logic, and the HOL object logic. Pure provides a term language,
a polymorphic type system, a syntax-translation framework for extensible parsing and pretty printing, and an inference
engine. The jEdit-based IDE allows \LaTeX-like term rendering using Unicode.

Isabelle theories consist of type declarations, definitions, and theorems. We prove theorems in Isabelle using
tactics. The simplifier tactic, \textsf{simp}, rewrites terms using equational theorems. The \textsf{auto} tactic
combines \textsf{simp} with deductive reasoning. Isabelle also has the powerful \textsf{sledgehammer} proof
tool~\cite{Blanchette2011} which invokes external first-order automated theorem provers on a proof goal, verifying their
results using tactics like \textsf{simp} and \textsf{metis}, which is a first-order resolution prover.

HOL implements a functional programming language founded on an axiomatic set theory. This object logic gives us a
principled approach to mechanised mathematics. We construct definitions and theorems by applying axioms in the proof
kernel. HOL provides inductive datatypes, recursive functions, and records. It provides basic types, including sets
($\power A$), total functions ($A \to B$), numbers ($\nat$, $\num$, $\mathbb{R}$), and lists. These types can be
parametric: $[nat]list$.\footnote{The square brackets are not used in Isabelle; we add them for readability.}
Specialisation unifies two types if one is an instantiation of the type variables of the other. For example, $[nat]list$
specialises $[\alpha]list$, where $\alpha$ is a type parameter.

\subsection{Verification and Semantic Embedding}
\label{sec:verify-sem}

In this section we consider different approaches to developing verification tools and outline the previous approaches to
UTP mechanisation~\cite{Nuka2006-UTP,Oliveira07,Feliachi2010,Zeyda2012,Foster14,Zeyda16} in this context.

Development of verification tools is usually conducted by means of a semantic embedding, where the language and
deductive reasoning laws are embedded in a proof assistant such as Isabelle/HOL or Coq. Building on Gordon's work
seminal work~\cite{Gordon1989-VerifyHOL}, Boulton et al.~\cite{BoultonGGHHT1992} identify the two fundamental categories
of semantic embedding: deep embeddings and shallow embeddings. In a deep embedding, the syntax tree of the language is
embedded into the host logic (such as HOL) as a datatype, and this acts as the basis for deductive verification calculi. In contrast,
in a shallow embedding, the syntax tree is implicit, and the goal is to directly reuse host logic reasoning facilities.

Boulton et al. note that deep embeddings allow reasoning over the syntactic structure of programs; for example we can
calculate the set of free variables in an expression. Consequently, a deep embedding can certainly support
Example~\ref{ex:asn-comm}, since syntax is first-class. The deep embedding technique is used, for example, by Nipkow and
Klein in their book \textit{Concrete Semantics}~\cite{Nipkow2014-ConcreteSemantics}. Nevertheless, deep embeddings have
the substantial disadvantage that they restrict the use of host logic proof facilities, which hampers efficient
verification. While we can mechanise \cref{ex:asn-comm}, it is not clear how we can efficiently apply it. Moreover, a
particular requirement for a UTP mechanisation is that the syntax tree is extensible, so that additional programming
operators can be defined, which is thwarted if we opt for a deep embedding.

\begin{figure}
  {\small
  \xymatrixcolsep{1pc}
  \xymatrix{
                                                 & & \text{\begin{tabular}{c}Zeyda et al. \\ (2012)~\cite{Zeyda2012} \end{tabular}} \ar[r] & \text{\begin{tabular}{c}Foster et al.  \\ (2014)~\cite{Foster14} \end{tabular}} \ar[r] & \text{\begin{tabular}{c} Zeyda et al. \\ (2016)~\cite{Zeyda16}\end{tabular}} \ar[rd] & & \\
    \text{\begin{tabular}{c}Nuka et al. \\ (2006)~\cite{Nuka2006-UTP}\end{tabular}} \ar@{.>}[r] & \text{\begin{tabular}{c}Oliveira et al. \\ (2007)~\cite{Oliveira07}\end{tabular}} \ar[ur] \ar@{.>}[dr] & & &                      & \text{\begin{tabular}{c}Foster et al. \\ (2016)~\cite{Foster16a}\end{tabular}} \\
                                                 & & \text{\begin{tabular}{c}Feliachi et al. \\ (2010)~\cite{Feliachi2010}\end{tabular}} \ar[rr]  &  & \text{\begin{tabular}{c}Feliachi et al. \\ (2012)~\cite{Feliachi2012} \end{tabular}} \ar[ru] &                     \\
 }
}

\caption{Overview of previous UTP semantic embeddings}
\label{fig:utp-mech-overview}

\end{figure}
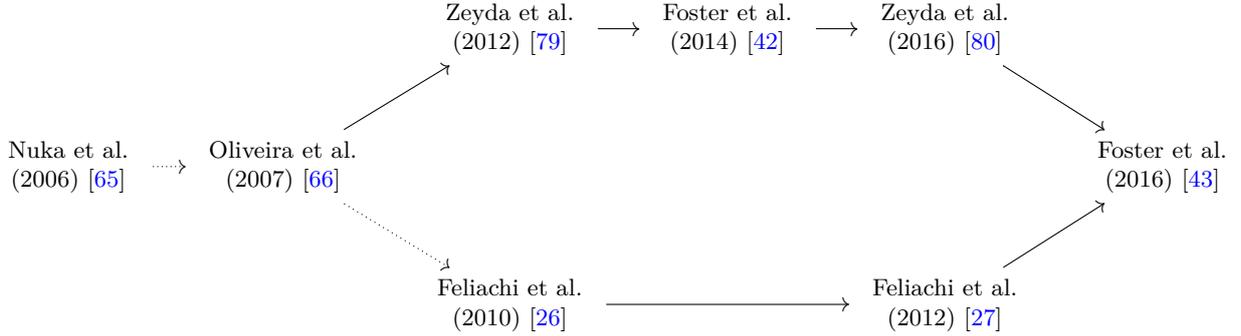

In contrast to deep embeddings, shallow embeddings have been very successful in supporting program
verification~\cite{Gordon1989-VerifyHOL,vonWright1994-RefHOL,Alkassar2008,Feliachi2010,Armstrong2015,Gomes2016}. As a
prominent example, the seL4 microkernel verification project uses a shallow embedding called Simpl~\cite{Alkassar2008},
which illustrates its scalability. Moreover, most of the previous UTP mechanisations are also shallow
embeddings~\cite{Oliveira07,Feliachi2010,Zeyda2012,Foster14}. The exceptions are Nuka and Woodcock~\cite{Nuka2006-UTP},
who follow the deep embedding approach, and Butterfield~\cite{Butterfield2010,Butterfield2012}, who develops a bespoke
proof tool called $U\cdot(TP)^2$. Within the shallow embeddings, broadly there are two streams, begun by Oliveira et
al.~\cite{Oliveira07,Zeyda2012,Foster14} and Feliachi et al.~\cite{Feliachi2010,Feliachi2012,Feliachi2015SymTest}, as
illustrated in Figure~\ref{fig:utp-mech-overview}. Both streams have HOL as the host logic, though Oliveira et
al.~\cite{Oliveira07} use a dialect called ProofPower-Z\footnote{ProofPower:
  \url{http://www.lemma-one.com/ProofPower/index/}}, whereas Feliachi et al.~\cite{Feliachi2010} use Isabelle/HOL. We
will consider further differences in \cref{sec:state-space-modelling}, but note that Isabelle/UTP~\cite{Foster16a}
results from the confluence of the ideas from the streams~\cite{Feliachi2010,Foster14}.

Isabelle/UTP is developed as a shallow embedding. All shallow embeddings follow roughly the same basic
approach~\cite{Gordon1989-VerifyHOL,vonWright1990-RefHOL}. First, we fix a state space $\src$, to describe states of a
program. Afterwards, we can model predicates using the type $\src \to \bool$, and programs using
$\src \times \src \to \bool$, which are binary relations. From this foundation, the programming language operators can
be defined using the predicative programming approach~\cite{Hehner1984-PredicativeProgramming}, where programs are
represented as below.

\begin{definition}[Programs as Predicates] \label{def:prog-as-pred}
\begin{align*}
  P \relsemi Q &~\defs~ (\lambda (s, s') @ (\exists s_0 @ P(s, s_0) \land Q(s_0, s'))) \\
  x := e       &~\defs~ (\lambda (s, s') @ s'.x = e(s) \land s'.y_1 = s.y_1 \land \cdots \land s'.y_n = s.y_n) \\
  \ckey{if}~b~\ckey{then}~P~\ckey{else}~Q~\ckey{fi} &~\defs~ (\lambda (s, s') @ (b(s) \land P(s, s')) \lor (\neg b(s) \land Q(s, s')))
\end{align*}
\end{definition}

\noindent These predicates effectively describe whether a given input-output pair, $(s, s') : \src \times \src$, is an
observation of the program. Sequential composition $P \relsemi Q$ is a predicative version of relational
composition~\cite{Tarski41}. It requires that there exists a middle state $s_0$ such that $P$ and $Q$ have this as a
possible final state and initial state, respectively. Assignment $x := e$ states that $x$ in the final state $s'$ has the value
$e$, and every other variable ($y_i$) retains its value. If-then-else conditional admits the behaviours of $P$ when $b$
is true, and the behaviours of $Q$ when $b$ is false. We can also denote the partial correctness Hoare triple
operator~\cite{Gordon1989-VerifyHOL}:
$$\hoaretriple{p}{Q}{r} \defs \left(\forall (s, s') @ p(s) \land Q(s, s') \implies r(s')\right)$$ A Hoare triple is
valid if, for any $(s, s')$ where the precondition $p$ is satisfied by $s$, and $Q$ has a final state $s'$ when started
from $s$, the postcondition $r$ is satisfied by $s'$. From this definition, we can prove many of the Hoare logic axioms
as theorems~\cite{Hoare69,Gordon1989-VerifyHOL}. Other axiomatic verification
calculi~\cite{Dijkstra75,Mor1996,Back1998,Hoare1994} can be characterised in a similar way -- this is the standard
shallow embedding approach.

However, not all laws are straightforward to express. In shallow embeddings, variables are often not first-class
citizens~\cite{vonWright1994-RefHOL,Alkassar2008,Feliachi2010}. That being the case, they are not objects of the logic,
and so it is not possible to express \cref{ex:asn-comm}. Moreover, consider the following variant of the forward
assignment law:
$$\hoaretriple{p}{x := e}{x = e \land p} \quad \text{provided } x \text { is not free in } p \text{ and } e$$
\noindent This is certainly a useful law, as it allows us, for example, to push initial assignments with constants
forward, such as $x := 1$. However, it requires that we can determine whether the variable $x$ is free in $p$. This is
seemingly a property that can only be expressed if $e$ and $p$ are syntactic objects. Example laws of this kind exist in
many axiomatic calculi~\cite{Dijkstra75,Mor1996,Back1998}. Now, to be clear, the absence of this law does not prevent a
particular program from being verified, and so it may not seem important. However, if the goal is unification by
characterising such laws abstractly, as is the case in UTP, then this is a significant question. Adequately answering
this is key to ensure that Isabelle/UTP is truly a unifying framework. In Definition~\ref{def:prog-as-pred}, we use the
notation $s.x$ to represent the value of variable $x$ in state $s$. How this operator is represented depends on how we
model state spaces, the crucial question for shallow embeddings, which we consider next.

\subsection{State Space Modelling}
\label{sec:state-space-modelling}

The principal difference between the shallow embeddings in \cref{sec:verify-sem} is their approach to modelling the
observation space $\src$. Schirmer and Wenzel provide a helpful discussion on modelling state
spaces~\cite{Schirmer2009}, and so we employ their framework for comparison and critique. They identify four common ways
of mechanising the modelling of state: using (1) functions; (2) tuples, (3) records, and (4) abstract types.

The first approach models state as a function, $\src \defs \textit{Var} \to \textit{Value}$, for suitable value and
variable types, and so $s.x \defs s(x)$. Gordon~\cite{Gordon1989-VerifyHOL}, Back and von
Wright~\cite{vonWright1990-RefHOL}, Oliveira et al.~\cite{Oliveira07}, and the successor UTP
mechanisations~\cite{Zeyda2012, Foster14, Zeyda16}, follow this approach. It requires a deep model of variables and
values. Consequently, it has similarities with deep embeddings, since concepts such as names and typing are first-class
citizens. This provides an expressive model with few limitations on possible manipulations of variables in the state
space~\cite{Foster14}. However, Schirmer and Wenzel highlight two obstacles~\cite{Schirmer2009}. First, the machinery
required for deep reasoning about values is heavy and \emph{a priori} limits possible value constructions, due to
cardinality restrictions in HOL. Second, explicit variable naming means the embedding must tackle syntactic issues, like
$\alpha$-renaming.

Zeyda et al.~\cite{Zeyda16} mitigate the first issue by axiomatically introducing a value universe ($\textit{Value}$) in
Isabelle. This universe has a higher cardinality that any HOL type, and so all normal types can be injected into it. The
cost of this approach is the extension of HOL with additional axioms. Moreover, the complexities associated with the
second issue remain. Once names and types are first-class citizens, it becomes necessary to replicate a large part of
the underlying meta-logic, such as a type checker. This requires great effort and can hit proof efficiency. Even so, the
functional state space approach seems necessary to model the dynamic creation of variables, as required, for example, in
modelling memory heaps in separation logic~\cite{Calcagno2007,Dongol2015}, so we do not reject it entirely.

The second approach~\cite{Schirmer2009} uses tuples to represent state; for example $\num \times \bool \times \num$ can
represent a state space with three variables~\cite{vonWright1994-RefHOL}. The value of each variable can be obtained by
decomposing the state using pattern matching, or using projection functions so that $s.x \defs \pi_n(s)$. The main issue
with this approach is that variable names are not automatically represented by the state space.

The third approach uses records to model state: a technique often used by verification tools in
Isabelle~\cite{Alkassar2008,Feliachi2010,Feliachi2012,Armstrong2015}. It is similar to the second approach, since
records are simply tuples. However, records come with bespoke selection and update functions for each index; this makes
manipulating the state space straightforward. Thus, we have $s.x \defs x(s)$, with $x$ being a field selector
function. Feliachi et al. use this approach to create their semantic embedding of the UTP in
Isabelle/HOL~\cite{Feliachi2010}. A record field represents a variable in this model. These can be abstractly
represented using pairs of field-query and update functions, $f_i$ and $f_i\textit{-upd}$. We do not need to encode the
set of variable names (\textit{Var}) in this approach.

The record state space approach greatly simplifies automation of program
verification~\cite{Feliachi2010,Feliachi2012,Feliachi2015SymTest,Armstrong2015,Gomes2016}. This is through directly
harnessing, rather than replicating, the polymorphic type system and automated proof tactics. The expense, though, is a
loss of flexibility compared to the functional approach, particularly in the decomposition of state spaces. Moreover,
a field of a record is not a first-class citizen because it is not an object of the logic in Isabelle. This means
that record fields lack the semantic structure necessary to capture their behaviour and thus manipulate or compare them
-- $f_i$ and $f_j$ are simply different functions. Consequently, implementations using records seldom provide general
support for syntactic concepts like free variables and substitution.

The previous approaches all use concrete models (types) for $\src$. The fourth approach~\cite{Schirmer2009} uses an
abstract type to represent a state space, with axiomatised projection functions for each of the variables. In this model
we have again that $s.x \defs x(s)$, but $x : \src \to \view$ is simply an abstract function without an explicit
implementation. This approach is employed by Back and von Wright~\cite{Back1998}, who use two functions
$\textsf{val}.x : \src \to \view$ and $\textsf{set}.x : \view \to \src \to \src$, to characterise each variable,
together with five axioms that characterise their behaviour (see \cref{def:back-var}). However, as indicated by both
Schirmer and Wenzel~\cite{Schirmer2009}, and Back and Preoteasa~\cite{Back2005-Procedures}, this approach requires us to
\textit{a priori} fix the number of variables available and their axioms. This hampers modularity in mechanisation,
since it is difficult to add new variables to grow a state space.

Schirmer and Wenzel's solution~\cite{Schirmer2009} is to adopt the state as functions approach, but improve its
flexibility using locales~\cite{Ballarin06}. An Isabelle locale allows the creation of a local theory context, with
fixed polymorphic constants and axiomatic laws~\cite{HOL-Algebra}. Rather than fixing concrete types for \textit{Var}
and \textit{Val}, Schirmer and Wenzel characterise these abstractly, and use locale constants to characterise injection
and projection functions. This is very similar to the approach adopted in our earlier version of
Isabelle/UTP~\cite{Foster14}, except that the latter uses type classes to assign injections to a polymorphic value
universe.

This approach imposes limitations that make it unsuitable for UTP, because it limits polymorphism in
variable types\footnote{This fact was first pointed out to us by Prof. Burkhart Wolff. Constants fixed in the head of an
  Isabelle locale have fixed types and cannot be polymorphic. In contrast, constants introduced at the global theory
  level can be fully polymorphic.}. For example, in the UTP theory of reactive processes, a trace variable $tr$ can be given
the polymorphic type $[e]\textit{list}$ for some event type $e$. When we hide events in a process, the event type can
change, since some events are no longer visible. Yet, in a locale, constants have a fixed type and so $tr$ is not truly
polymorphic. In comparison, if we define a record with a field $tr : [e]\textit{list}$ then we can assign it different
types. We conclude that there is a need for a different approach to state space modelling.

\subsection{Our Approach}
\label{sec:our-approach}

Our approach is closest to the abstract type approach, though we aim is to unify all four. For UTP, we need to treat
variables as first-class citizens. We see merits both in modelling state as a function, and also as a record. Indeed, we
recognise that there is often a need to blend these two representations, as we illustrate using the running example
below:

\begin{example}[Stores Variables and Heaps] \label{ex:st-hp}
Consider a state space with three variables $x : int$, $y : int$, and $hp : addr \to int$. The variables $x$ and $y$
represent program variables in the store, and $hp$ represents a heap mapping addresses to integer values. It can be
modelled as a record with three fields. However, we will likely want to perform assignments directly to elements of
function $hp$, for example $hp[loc] := 7$, and so we see the two state representations, functions and records,
co-existing. \qed
\end{example}

Isabelle/UTP generalises the various state-space approaches by abstractly characterising variables algebraically using
lenses~\cite{Foster07,Foster09,Hofmann2011-SymLens,Pickering2017-Optics}. A lens consists of two functions:
$\lget : \src \to \view$ that extracts a value from a state, and $\lput : \src \to \view \to \src$ that puts back an
updated view. We characterise lenses as algebraic structures to which different concrete models can be assigned, which
allows us to unify the various state space approaches. Lenses allow characterisation of both individual variables and
also state space \textit{regions} that can encompass several variables. We consider, for example, that the heap location
$hp[loc]$ is semantically a part of the heap $hp$, and therefore changes to $hp$ also effect $hp[loc]$. Moreover, in an
object oriented program the state is hierarchical: the attribute variables are all part of the object variable.

Since lenses can characterise sets of variables, we can also use them to model frames, as required, for example, by
Morgan's refinement calculus~\cite{Mor1996}. We define operators that allow us to compare and manipulate lenses,
including independence ($x \bowtie y$), sublens ($x \lsubseteq y$), and summation ($x \lplus y$), which effectively
allows execution of two lenses in parallel, but differently to Pickering's operator~\cite{Pickering2017-Optics}, which
acts on a product space. We also implement N.~Foster's lens composition operator ($x \lcomp y$)~\cite{Foster09}, which
supports hierarchical state spaces.

Isabelle/UTP, therefore, has a high level of proof automation because it is a shallow
embedding~\cite{vonWright1994-RefHOL,Feliachi2010,Feliachi2012}, and we avoid the requirement to explicitly characterise
names and values. Like any other object in Isabelle, we can assign a name to a particular lens, but this name is
meta-logical. Such globally named lenses can also be polymorphic, which helps us to model observational variables. At
the same time, even though Isabelle/UTP is a shallow embedding, the lens axioms provide us with sufficient structure to
characterise syntax-like queries, like substitution and free variables. Consequently, lenses allow us to develop a
program model that exhibits benefits of both deep and shallow embeddings.

\section{Algebraic Observation Spaces}
\label{sec:lenses}
In this section, we present our theory of lenses, which provides an algebraic semantics for state and observation space
modelling in Isabelle/UTP. Although some core definitions like the lens laws and composition operator are well
known~\cite{Foster07,Foster09,Fischer2015}, we introduce several novel operators, including summation, and relations
like independence and equivalence. We prove several algebraic properties for these operators, which are foundational for
our mechanisation of UTP.

All definitions, theorems, and proofs in this section may be found in our Isabelle/HOL
mechanisation~\cite{Foster18c-Optics}.

\subsection{Signature}

A lens is used, intuitively, to view and manipulate a region ($\view$) of a state space ($\src$), as illustrated in
Figure~\ref{fig:lens}. The view, $\view$, corresponds to the hatched region of $\src$. A region may model the contents
of one or more variables, whose type is $\view$. We introduce lenses as two-sorted algebraic structures.
\begin{definition}[Lenses] 
  A lens is a quadruple $\lquad{\view}{\src}{\lget : \src \to \view}{\lput : \src \to \view \to \src}$, where $\view$
  and $\src$ are non-empty sets called the view type and state space\footnote{In the lens
    literature~\cite{Foster09,Fischer2015,Pickering2017-Optics} this is referred to as the ``source''. We refer to it as
    the state space and observation space, a more general concept, interchangeably, depending on the context.},
  respectively, and $\lget$ and $\lput$ are total functions\footnote{N.~Foster~\cite{Foster09} also introduces a
    function called $\lcreate : \view \to \src$ that creates an element of the source. We omit this because it can be
    defined in terms of $\lput$ and is not necessary for this paper.}. The view $\view$ corresponds to a region of the
  state space $\src$. We write $\view \lto \src$ to denote the type of lenses with state space $\src$ and view type
  $\view$, and subscript $\lget$ and $\lput$ with the name of a particular lens.
  \isalink{https://github.com/isabelle-utp/utp-main/blob/d18036c259bbc47cbd4200f2428258934e935699/optics/Lens_Laws.thy\#L13}
\end{definition}
\noindent This follows the standard definition given by N.~Foster~\cite{Foster09}, though other works~\cite{Foster07}
employ partial functions. The $\lget$ function views the current value of the region, and $\lput$ updates
it. Intuitively, we use these structures to model sequences of queries and updates on the state space $\src$ in
\S\ref{sec:isabelle-utp}. Each variable in a program can be represented by an individual lens, with operators like
assignment utilising $\lget$ and $\lput$ to manipulate the corresponding region of memory. For the purpose of
example, we describe lenses for record types.

\begin{definition}[Record Field Lens] 
  We consider the definition of a new record type,
  $R \defs \llparenthesis f_1 : \tau_1, \cdots f_n : \tau_n \rrparenthesis$, with $n$ fields, each having a
  corresponding type. Each field yields a function $f_i : R \to \tau_i$, which queries the current value of a
  field. Moreover, we can update the value of field $f_i$ in $r : R$ with $k : \tau_i$ using $r(f_i := k)$. We can
  construct a lens for each field using $\lrec^{R}_{f_i} \defs \lquad{\tau_i}{R}{f_i}{\lambda s~v @ s(f_i := v)}$. \isalink{https://github.com/isabelle-utp/utp-main/blob/d18036c259bbc47cbd4200f2428258934e935699/optics/Lens_Instances.thy\#L207}
\end{definition}

\noindent The field lens allows us to employ the ``state as records'' approach~\cite{Schirmer2009,Feliachi2010}, as
discussed in \cref{sec:state-space-modelling}. We also consider a second example. Many state spaces are built using the
Cartesian product type $S_1 \times S_2$, and consequently it is useful to define lenses for such a space. We therefore
define the $\lfst$ and $\lsnd$ lenses.

\begin{definition}[Product Projection Lenses] \isalink{https://github.com/isabelle-utp/utp-main/blob/d18036c259bbc47cbd4200f2428258934e935699/optics/Lens_Algebra.thy\#L49}
  \begin{align*}
  \lfst^{S_1,S_2} ~~\defs~~& \lquad{S_1}{S_1 \times S_2}{\lambda (x, y) @ x}{\lambda (x, y)~z @ (z, y)} \\
  \lsnd^{S_1,S_2} ~~\defs~~& \lquad{S_2}{S_1 \times S_2}{\lambda (x, y) @ y}{\lambda (x, y)~z @ (x, z)}
  \end{align*}
\end{definition}

\noindent The superscripted state spaces are necessary in order to specify the product type; we omit them when they
can be determined from the context. Lenses $\lfst$ and $\lsnd$ allow us to focus on the first and second element of a
product type, respectively. Their $\lget$ functions project out these elements, and the $\lput$ functions replace the
first and second elements with the given value $z$. An application of these lenses is the ``states as products''
approach~\cite{Schirmer2009} (see \cref{sec:state-space-modelling}): we can directly model a state space with two
variables, for example $x \defs \lfst^{\num, \bool}$ and $y \defs \lsnd^{\num, \bool}$ give a state space with
$x : \num$ and $y : \bool$. A final example is the total function lens.

\begin{definition} $\lfun^{A,B}_k \defs \lquad{B}{A \to B}{\lambda f @ f(k)}{\lambda f~v @ (\lambda x @ \hifthenelse{x = k}{v}{f(x)})}$ \isalink{https://github.com/isabelle-utp/utp-main/blob/d18036c259bbc47cbd4200f2428258934e935699/optics/Lens_Instances.thy\#L11}
\end{definition}

\noindent The total function lens $\lfun^{A,B}_k$ views the output of a function $f : A \to B$ associated with a given input value
$k : A$. The $\lget$ function simply applies $f$ to $k$, and the $\lput$ function associates a new output $v$ with
$k$. The function lens allows us to also employ the ``state as functions'' approach~\cite{Oliveira07,Schirmer2009}. We
can also use it to model an array of integer values with $\lfun^{\nat,\num}_k$, as used in \cref{ex:st-hp}.

\begin{figure}
  \centering
  \begin{minipage}{.25\linewidth}
    \includegraphics[width=\linewidth]{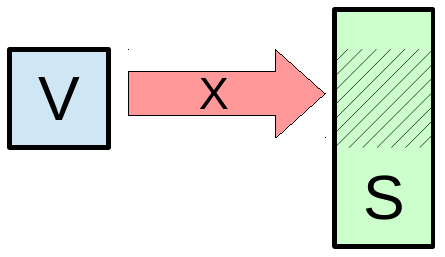}
    \caption{Lens viewing a region}
    \label{fig:lens}
  \end{minipage}\quad\begin{minipage}{.25\linewidth}
    \includegraphics[width=\linewidth]{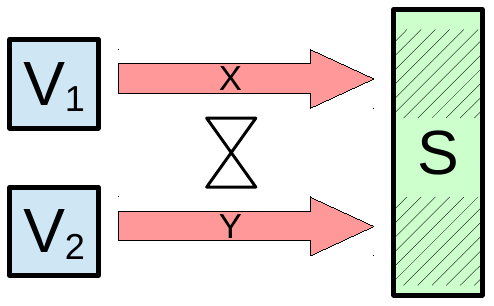}
    \caption{Independence}
    \label{fig:indep}
  \end{minipage}\quad\begin{minipage}{.4\linewidth}
    \includegraphics[width=\linewidth]{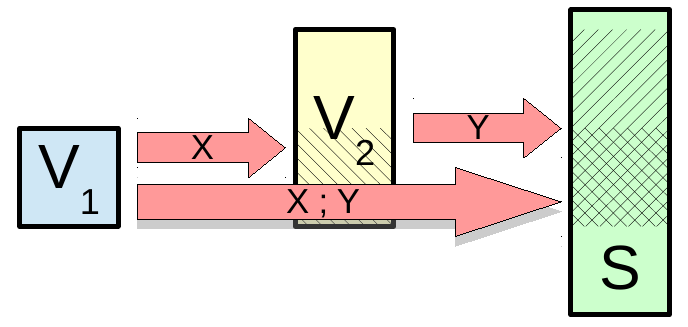}

    \vspace{-3.2ex}
    \caption{Lens composition}
    \label{fig:lcomp}
  \end{minipage}

\end{figure}

\subsection{Axiomatic Basis}
\label{sec:lensax}

The use of lenses to model variables depends on $\lget$ and $\lput$ behaving according to a set of axioms.
\begin{definition}[Total Lenses] \label{def:total-lenses}
  A total lens obeys the following axioms: \isalink{https://github.com/isabelle-utp/utp-main/blob/d18036c259bbc47cbd4200f2428258934e935699/optics/Lens_Laws.thy\#L174}
  \begin{align*}
    \lget~(\lput~s~v) &= v \tag{PutGet} \label{law:put-get} \\
    \lput~(\lput~s~v')~v &= \lput~s~v \tag{PutPut} \label{law:put-put} \\
    \lput~s~(\lget~s) &= s \tag{GetPut} \label{law:get-put}
  \end{align*}
  We write $\view \ltto \src$ for the set of total lenses with view type $\view$ and state space $\src$, and
  $\ast \ltto \src$ for the set of total lenses with any view type, whose state space is $\src$.
\end{definition}
We mechanise this algebraic structure in Isabelle using locales, following the pattern given by
Ballarin~\cite{Ballarin06}\footnote{We are not using locales to characterise state spaces, like Schirmer and
  Wenzel~\cite{Schirmer2009}, but simply to fix the algebra of lenses.}. Total lenses are usually called ``very
well-behaved'' lenses~\cite{Foster07,Foster09}, but we believe ``total'' is more descriptive, since it is always
possible to meaningfully project a view from a state. Axiom~\ref{law:put-get} states that if a state has been constructed by application of
$\lput~s~v$, then a matching $\lget$ returns the injected value, $v$. Axiom~\ref{law:put-put} states that a later
$\lput$ overrides an earlier one, so that the previously injected value $v'$ is replaced by $v$. Finally,
Axiom~\ref{law:get-put} states that for any state element $s$, if we extract the view element and then update the
original using it, then we get precisely $s$ back.

We now demonstrate that every field of a record yields a total lens.

\begin{lemma} \label{thm:reclens} For any field $f_i$ of a record type $R$, the record lens 
  $\lrec^{R}_{f_i}$ forms a total lens.
\end{lemma}
\begin{proof}
For illustration, we prove each lens axiom in turn.
\begin{enumerate}
  \item \ref{law:put-get}: $\lget~(\lput~s~v) = f_i~(s(f_i := v)) = v$
  \item \ref{law:put-put}: $\lput~(\lput~s~v')~v = (\lambda s~v @ s(f_i := v))~(s(f_i := v'))~v = (\lambda v @ s(f_i := v')(f_i := v))~v = \lput~s~v$
  \item \ref{law:get-put}: $\lput~s~(\lget~s) = (\lambda s~v @ s(f_i := v))~s~(f_i~s) = s(f_i := f_i~s) = s$
\end{enumerate}
Consequently, the record lens is indeed a total lens.
\end{proof}

We can similarly show that $\lfst$, $\lsnd$, and $\lfun$ are total lenses\footnote{All omitted proofs can be found in
  our Isabelle/HOL mechanisation~\cite{Foster18c-Optics}.}.

\begin{lemma} For any $A$,$B$, and $k \in A$, $\lfst^{A,B}$, $\lsnd^{A,B}$, and $\lfun^{A,B}_k$ are total lenses. \isalink{https://github.com/isabelle-utp/utp-main/blob/d18036c259bbc47cbd4200f2428258934e935699/optics/Lens_Instances.thy\#L20}
\end{lemma}

\noindent While both \ref{law:put-get} and \ref{law:put-put} are satisfied for most useful state-space models we can
consider, this is not the case for \ref{law:get-put}. For example, if we consider a lens that projects the valuation of
an element $x : A$ from a partial function $f : A \pfun B$, then $\lget$ is only meaningful when $x \in \dom(f)$. Since
$\lget$ is total, it must return a value, but this will be arbitrary and therefore placing it back into $f$ alters its
domain. We do not consider lenses that do not satisfy \ref{law:get-put} in this paper, but simply observe that total
lenses are a useful, though not universal solution for state space modelling.

\subsection{Independence}
\label{sec:lens-indep}

So far we have considered the behaviour of individual lenses, but programs reference several variables and so it is
necessary to compare them. One of the most important relationships between lenses is independence of their corresponding
views, which is illustrated in Figure~\ref{fig:indep}. We formally characterise independence below.
\begin{definition}[Independent Lenses] \label{def:lens-indep}
\noindent Lenses $X : \view_1 \lto \src$ and $Y : \view_2 \lto \src$ are independent, written $X \lindep Y$, provided they
satisfy the following laws: \isalink{https://github.com/isabelle-utp/utp-main/blob/d18036c259bbc47cbd4200f2428258934e935699/optics/Lens_Laws.thy\#L326}
  \vspace{-1ex}
  \begin{align*}
    \lput_X~(\lput_Y~s~v)~u &~=~ \lput_Y~(\lput_X~s~u)~v \tag{LI1} \label{law:PutsComm} \\
    \lget_X~(\lput_Y~s~v) &~=~ \lget_X~s \tag{LI2} \label{law:PutIrr1} \\
    \lget_Y~(\lput_X~s~u) &~=~ \lget_Y~s \tag{LI3} \label{law:PutIrr2}
  \end{align*}
\end{definition}
\noindent Lenses $X$ and $Y$, which share the same state space but not necessarily the same view type, are independent
provided that applications of their respective $\lput$ functions commute (\ref{law:PutsComm}), and their respective
$\lget$ functions are not influenced by the corresponding $\lput$ functions (\ref{law:PutIrr1}, \ref{law:PutIrr2}). In
the encoding of \cref{ex:st-hp}, we have that $x \lindep y$, $x \lindep hp$, and $y \lindep hp$: these are distinct
variables. Nevertheless, lens independence captures a deeper concept, since lenses with different types are not
guaranteed to be independent, as they might represent a different view on the same region.

Independence can, for example, be used capture the condition for commutativity of assignments. If $x$ and $y$ are
variables, then we can characterise the following law
$$(x := e \relsemi y := f) = (y := f \relsemi x := e) \text {\quad provided $x \lindep y$, and $e$ and $f$ are constants}$$
which partly formalises \cref{ex:asn-comm}. Such assignment laws will be explored further in
Section~\ref{sec:isabelle-utp}. We can show that $\lfst \lindep \lsnd$ using the calculation below.

\begin{lemma} $\lfst \lindep \lsnd$ \isalink{https://github.com/isabelle-utp/utp-main/blob/master/optics/Lens_Algebra.thy\#L301} \end{lemma}
\begin{proof} We first prove that $\lput_\lfst$ commutes with $\lput_\lsnd$ (\ref{law:PutsComm}) by evaluation, assuming $s = (s_1, s_2)$:
  \begin{align*} & \lput_\lfst (\lput_\lsnd~s~v)~u ~=~ \lput_\lfst (\lput_\lsnd~(s_1, s_2)~v)~u \\
    =~~& \lput_\lfst~(s_1, v)~u ~=~ (u, v) ~=~ \lput_\lsnd~(u, s_2)~v \\ =~~& \lput_\lsnd~(\lput_\lfst~(s_1, s_2)~u)~v
    ~=~ \lput_\lsnd~(\lput_\lfst~s~u)~v \end{align*} Similarly, we can also prove \ref{law:PutIrr1} by evaluation:
  $$\lget_\lfst~(\lput_\lsnd~s~v) ~=~ \lget_\lfst~(\lput_\lsnd~(s_1, s_2)~v) ~=~ \lget_\lfst~(s_1, v) ~=~ s_1 ~=~ \lget_\lfst~(s_1,
  s_2) ~=~ \lget_\lfst~s$$ Finally, \ref{law:PutIrr2} follows by a symmetric proof.
\end{proof}

A further, illustrative result is the meaning of independence in the total function lens:

\begin{lemma} \label{thm:lfun-indep} $\lfun_a \lindep \lfun_b \iff a \neq b$ \isalink{https://github.com/isabelle-utp/utp-main/blob/d18036c259bbc47cbd4200f2428258934e935699/optics/Lens_Instances.thy\#L25}
\end{lemma} 

\noindent Two instances of the total function lens are independent if, and only if, the parametrised inputs $a$ and $b$ are
different. This reflects the intuition of a function -- every input is associated with a distinct output.

We can actually show that axioms \ref{law:PutIrr1} and \ref{law:PutIrr2} in Definition~\ref{def:lens-indep} can be
dispensed with when $X$ and $Y$ are both total lenses. Consequently, we can adopt a simpler definition of independence:

\begin{theorem} \label{thm:indep-total}
  If lenses $X : \view_1 \lto \src$ and $Y : \view_2 \lto \src$ are both total then \isalink{https://github.com/isabelle-utp/utp-main/blob/d18036c259bbc47cbd4200f2428258934e935699/optics/Lens_Laws.thy\#L356}
  $$X \lindep Y \iff \forall (u, v, s) @ \lput_X~(\lput_Y~s~v)~u ~=~ \lput_Y~(\lput_X~s~u)~v$$
\end{theorem}

\noindent However, the weaker definition of independence is still useful for the situation when not all three axioms of
total lenses are satisfied~\cite{Foster20-LocalVars}.

\subsection{Lens Combinators}

Lenses can be independent, but they can also be ordered by containment using the sublens relation $X \lsubseteq Y$,
which orders lenses. The intuition is that $Y$ captures a larger region of the state space than $X$. One interpretation
is a subset operator for relating lens sets. Before we can get to the definition of this $\lsubseteq$, we first need to
define some basic lens combinators.
\begin{definition}[Basic Lenses] $\lzero_S ~\defs~ \lquad{\{\emptyset\}}{S}{\lambda s @ \emptyset}{\lambda s~v @ s} \qquad \lone_S  ~\defs~ \lquad{S}{S}{\lambda s @ s}{\lambda s~v @ v}$ \isalink{https://github.com/isabelle-utp/utp-main/blob/d18036c259bbc47cbd4200f2428258934e935699/optics/Lens_Algebra.thy\#L69}
\end{definition}
\begin{lemma}[Basic Lenses Closure] For any $S$, $\lzero_S$ and $\lone_S$ are total lenses. \isalink{https://github.com/isabelle-utp/utp-main/blob/d18036c259bbc47cbd4200f2428258934e935699/optics/Lens_Algebra.thy\#L126} \end{lemma}
\noindent The $\lzero$ lens has a unitary view type, $\{\emptyset\}$. Consequently, for any element of the
state space, it always views the same value $\emptyset$. It cannot be used to either observe or change a state, and it
is therefore entirely ineffectual in nature. It can be interpreted as the empty set of lenses. Conversely, the $\lone$
lens, with type $S \lto S$, views the entirety of the state. It can be interpreted as the set of all lenses in the
state space. We sometimes omit type information from these basic lenses when this can be inferred from the context.

It is useful in several circumstances to chain lenses together, provided that the view of one matches the source of the
other. For this, we adopt the lens composition operator, originally defined by J.~Foster~\cite{Foster09}.
\begin{definition} \label{def:lenscomp} $X \lcomp Y ~~\defs~~ \lquad{\view_X}{\src_Y}{\lget_X \circ \lget_Y}{\lambda s~v @ \lput_Y~s~(\lput_X~(\lget_Y~s)~v)}  ~~ \text{if } \src_X = \view_Y$ \isalink{https://github.com/isabelle-utp/utp-main/blob/d18036c259bbc47cbd4200f2428258934e935699/optics/Lens_Algebra.thy\#L22}
  
\end{definition}
\begin{lemma}[Composition Closure] If $X$ and $Y$ are total lenses, then $X \lcomp Y$ is a total lens. \isalink{https://github.com/isabelle-utp/utp-main/blob/d18036c259bbc47cbd4200f2428258934e935699/optics/Lens_Algebra.thy\#L132}
\end{lemma}
\noindent A lens composition, $X \lcomp Y$, for $X : A \lto B$ and $Y : B \lto C$ chains together two lenses. It is illustrated in
Figure~\ref{fig:lcomp}. When $X$ characterises an $A$-shaped region of $B$, and $Y$ characterises a $B$-shaped region of
$C$, overall lens $X \lcomp Y$ characterises an $A$-shaped region of $C$. It is useful when we have a state space
composed of several individual state components, and we wish to select a variable of an individual component.
\begin{example} \label{ex:objlens} In an object oriented program we may have $m \in \nat$
  objects whose states are characterised by lenses $o_i : O_i \lto \src$, for $i \in \{1..m\}$, where each $O_i$
  characterises the respective object state. An object $o_k$ has $n \in \nat$ attributes, characterised by lenses $x_j
  : \tau_j \lto O_k$ for $j \in \{1..n\}$. We can select one of these attributes from the global state context by the
  composition $x_j \lcomp o_k : \tau_j \lto \src$. \qed
\end{example}
\noindent Composition is also useful for collections, such as the heap array in \cref{ex:st-hp}.
If $hp : (addr \to \num) \lto \src$, that is, a lens which views a function in $\src$, then we can represent a lookup,
$hp[loc]$, by the composition $\lfun^{addr,\num}_{loc} \lcomp hp$\footnote{Here, $loc$ is a constant value, though this
can be relaxed to an expression that can depend on other state variables~\cite{Foster20-LocalVars}.}.

Lens composition obeys a number of useful algebraic properties, as shown below.
\begin{theorem}[Composition Laws] If $X : A \lto B$, $Y : B \lto C$, and $Z : C \lto D$ are total lenses then the following identities hold: \isalink{https://github.com/isabelle-utp/utp-main/blob/d18036c259bbc47cbd4200f2428258934e935699/optics/Lens_Algebra.thy\#L214}
  $$X \lcomp (Y \lcomp Z) ~~=~~ (X \lcomp Y) \lcomp Z \qquad X \lcomp \lone_B ~~=~~ \lone_A \lcomp X = X \qquad \lzero_A \lcomp X ~~=~~ \lzero_B$$
\end{theorem}
\noindent Lens composition is associative, since the order in which lenses are composed is irrelevant, and it has $\lone$ as its left and right units. Moreover, $\lzero$ is a left
annihilator, since if the view is reduced to $\{\emptyset\}$ then no further data can be extracted.

While composition can be used to chain lenses in sequence, it is also possible to compose them in parallel, which is the
purpose of the next operator.
\begin{definition}[Lens Sum] \label{def:lens-sum} \isalink{https://github.com/isabelle-utp/utp-main/blob/d18036c259bbc47cbd4200f2428258934e935699/optics/Lens_Algebra.thy\#L38}
  \begin{align*}
    X \lplus Y ~~\defs~~& \lquad{\view_X \times \view_Y}{\src_X}{\lambda s @ (\lget_X~s, \lget_Y~s)}{\lambda s~(v_1, v_2) @ \lput_Y~(\lput_X~s~v_1)~v_2} & \text{if } \src_X = \src_Y
  \end{align*}
\end{definition}
\noindent Lens sum allows us to simultaneously manipulate two regions of $\src$, characterised by
$X : \view_1 \lto \src$ and $Y : \view_2 \lto \src$. Consequently, the view type is the product of the two constituent
views: $\view_1 \times \view_2$. The $\lget$ function applies both constituent $\lget$ functions in parallel. The
$\lput$ function applies the constituent $\lput$ functions, but in sequence since we are in the function domain. We can
prove that total independent lenses are closed under lens sum.

\begin{lemma}[Sum Closure] If $X$ and $Y$ are independent total lenses, then $X \lplus Y$ is a total lens. \isalink{https://github.com/isabelle-utp/utp-main/blob/d18036c259bbc47cbd4200f2428258934e935699/optics/Lens_Algebra.thy\#L157}
\end{lemma}
\noindent We require that $X \lindep Y$, since manipulation of two overlapping regions could have unexpected results,
and then also the order of the $\lput$ functions is irrelevant.

Lens sum can characterise independent concurrent views and updates to
the state space. For example, we can encode a simultaneous update to two variables as $(x + y) := (e, f)$. Moreover,
lens sum can also be used to characterise sets of independent lenses. If we model three variables using lenses $x$, $y$,
and $z$ which share the same source and are all independent, then the set $\{x, y, z\}$ can be represented by
$x + y + z$.

With the help of lens composition, we can now also prove some algebraic laws for lens sum.
\begin{lemma} If $X$ and $Y$ are independent total lenses then the following identities hold:
  \isalink{https://github.com/isabelle-utp/utp-main/blob/d18036c259bbc47cbd4200f2428258934e935699/optics/Lens_Algebra.thy\#L348}
  $$\lfst \lcomp (X \lplus Y) ~=~ X \quad \lsnd \lcomp (X \lplus Y) ~=~ Y \quad (X \lplus Y) \lcomp Z ~=~ (X \lcomp Z) \lplus (Y \lcomp Z)$$
\end{lemma}
\noindent The first two identities show that $\lfst$ and $\lsnd$ composed with $\lplus$ yield the left- and right-hand side
lenses, respectively. If we perform a simultaneous update using $X + Y$, but then throw away one them by composing with
$\lfst$ and $\lsnd$, then the result is simply $X$ or $Y$, respectively. The third identity shows that $\lcomp$
distributes from the right through lens sum. It does not in general distribute from the left as such a construction is
not well-formed. We next show some independence properties for the operators introduced so far.
\begin{lemma}[Independence] If $X$, $Y$, and $Z$ are total lenses then the following laws hold: \isalink{https://github.com/isabelle-utp/utp-main/blob/d18036c259bbc47cbd4200f2428258934e935699/optics/Lens_Algebra.thy\#L295}

  \begin{minipage}{.3\linewidth}
  \begin{align*}
    \lzero \lindep X \\
    X \lindep Y ~\iff~&Y \lindep X \\
    Y \lindep Z ~\implies~& (X \lcomp Y) \lindep Z
  \end{align*} 
  \end{minipage}
  \begin{minipage}{.7\linewidth}
  \begin{align*}
    (X \lcomp Z) \lindep (Y \lcomp Z) ~\iff~& X \lindep Y \\
    X \lindep Z \land Y \lindep Z ~\implies~& (X \lplus Y) \lindep Z
  \end{align*} 
  \end{minipage}
\end{lemma}
\noindent The $\lzero$ lens is independent from any lens, since it views none of the state space. Lens independence is a
symmetric relation, as expected. Lens composition preserves independence: if $Y \lindep Z$ then composing $X$ with $Y$
still yields a lens independent of $Z$. Referring back to \cref{ex:st-hp}, if $x \lindep hp$, then clearly also
$x \lindep hp[loc]$. $X$ and $Y$ composed with a common lens $Z$ are independent if, and only if, $X$ and $Y$ are
themselves independent. Combining this with \cref{thm:lfun-indep}, we can deduce that $hp[l_1]$ and $hp[l_2]$ are
independent if and only $l_1 \neq l_2$. Finally, lens sum also preserves independence: if both $X$ and $Y$ are
independent of $Z$, then also $X \lplus Y$ is independent of $Z$. Thus, if $x \lindep hp$ and $y \lindep hp$, then also
$x + y \lindep hp$.

\subsection{Observational Order and Equivalence} \label{sec:order-eq}
We recall that a lens $X$ can view a larger region than another lens $Y$, with the implication that $Y$ is fully
dependent on $X$. For example, it is clear that in \cref{ex:objlens} each object fully possesses each of its
attributes, and likewise each attribute lens $x_j \lcomp o_k$ is fully dependent on lens $o_k$. We can formalise this
using the lens order, $X \lsubseteq Y$, which we can now finally define.
\begin{definition}[Lens Order] \label{def:sublens} $(X \lsubseteq Y) \defs (\src_X = \src_Y \land (\exists Z : \view_X \ltto \view_Y @ X = Z \lcomp Y))$ \isalink{https://github.com/isabelle-utp/utp-main/blob/d18036c259bbc47cbd4200f2428258934e935699/optics/Lens_Order.thy\#L12}
\end{definition}
\noindent A lens $X : V_1 \lto S$ is narrower than a lens $Y : V_2 \lto S$ provided that they share the same state
space, and there exists a total lens $Z : V_1 \ltto V_2$, such that $X$ is the same as $Z \lcomp Y$. In other words, the
behaviour of $X$ is defined by firstly viewing the state using $Y$, and secondly viewing a subregion of this using
$Z$. An order characterises the size of a lens's aperture: how much of the state space a lens can view. For example, we
can prove that $x_j \lcomp o_k \lsubseteq o_k$, by setting $Z = x_j$ in Definition~\ref{def:sublens}. For the same
reason, we can also show that $hp[loc] \lsubseteq hp$. The lens order relation is a preorder, as demonstrated below.
\begin{theorem} \label{thm:lpreorder} For any $\src$, $(\ast \ltto \src, \lsubseteq)$ forms a preorder, that is, $\lsubseteq$ is reflexive
  and transitive. Furthermore, the least element of $\ast \ltto \src$ is $\lzero_\src$, and the greatest element is
  $\lone_\src$. \isalink{https://github.com/isabelle-utp/utp-main/blob/d18036c259bbc47cbd4200f2428258934e935699/optics/Lens_Order.thy\#L29}
\end{theorem}
\noindent Clearly, $\lzero$ is the narrowest possible lens since it allows us to view nothing, and $\lone$ is the widest lens,
since it views the entire state. This is consistent with the intuition that $\lzero$ represents the empty set, $\lone$
represents the set of all lenses, and $\lsubseteq$ is a subset-like operator. We can prove the following intuitive
theorem for sublenses.
\begin{lemma} If $X, Y$ are total lenses and $X \lsubseteq Y$, then the following identities hold: \isalink{https://github.com/isabelle-utp/utp-main/blob/d18036c259bbc47cbd4200f2428258934e935699/optics/Lens_Order.thy\#L54}
  \begin{align*}
    \lput_Y~(\lput_X~s~v)~u &= \lput_Y~s~u \tag{LS1} \label{law:sl-put-put} \\
    \lget_X~(\lput_Y~s~v)   &= \lget_X~(\lput_Y~s'~v) \tag{LS2} \label{law:sl-obs-get}
  \end{align*}
\end{lemma}
\noindent Law~\ref{law:sl-put-put} is a generalisation of Axiom~\ref{law:put-put}: a later $\lput_Y$ overrides an earlier
$\lput_X$ when $X \lsubseteq Y$. Law~\ref{law:sl-obs-get} states that when viewing an update on $Y$ via a narrower lens $X$
we can ignore the valuation of the original state, since the update replaces all the relevant information. We can now
use these results to prove a number of ordering lemmas for lens compositions.
\begin{lemma}[Lens Order] If $X$, $Y$, and $Z$ are total lenses then they satisfy the following laws: \isalink{https://github.com/isabelle-utp/utp-main/blob/d18036c259bbc47cbd4200f2428258934e935699/optics/Lens_Order.thy\#L204}
  \begin{align*}
    X \lcomp Y &\lsubseteq Y \tag{LO1} \label{law:LO1} \\
    X \lsubseteq Y \land Y \lindep Z &\implies X \lindep Z \tag{LO2} \label{law:LO2} \\
    X \lindep Z \land Y \lsubseteq Z &\implies (X \lplus Y) \lsubseteq (X \lplus Z) \tag{LO3} \label{law:LO3} \\
    X \lindep Y &\implies X \lplus Y \lsubseteq Y \lplus X \tag{LO4} \label{law:LO4} \\
    X \lindep Y \land X \lindep Z \land Y \lindep Z &\implies X \lplus (Y \lplus Z) \lsubseteq (X \lplus Y) \lplus Z \tag{LO5} \label{law:LO5}
  \end{align*}
\end{lemma}
\noindent As we observed, composition of $X$ and $Y$ yields a narrower lens than $Y$ (\ref{law:LO1}):
$hp[loc] \lsubseteq hp$. Independence is preserved by the ordering, since a subregion of a larger independent region is
also clearly independent (\ref{law:LO2}) -- $hp[loc] \lindep x$ when $hp \lindep x$. Lens sum also preserves the
ordering in its right-hand component (\ref{law:LO3}). Moreover, lens sum is commutative with respect to $\lsubseteq$
(\ref{law:LO4}), and also associative, assuming appropriate independence properties (\ref{law:LO5}). From these laws,
and utilising the preorder theorems, we can prove various useful corollaries, such as
$$X \lindep Y \implies X \lsubseteq (X \lplus Y)$$ which shows that $X$ is narrower than $X \lplus Y$: the latter is an
upper bound. Thus, if we intuitively interpret $\lsubseteq$ as $\subseteq$, then $+$ corresponds to $\cup$, and we can
combine independent lens sets: $\{a, b\} \cup \{c, d\} = (a + b) + (c + d)$. We can also show that
$$X \lsubseteq Z \land Y \lsubseteq Z \land X \lindep Y \implies (X \lplus Y) \lsubseteq Z$$ which shows that sum
provides the least upper bound: $\cup$ preserves $\subseteq$. Finally, we can also induce an equivalence relation on
lenses using the lens order in the usual way:

\begin{definition}[Lens Equivalence] $(X \lequiv Y) \defs (X \lsubseteq Y \land Y \lsubseteq X)$ \isalink{https://github.com/isabelle-utp/utp-main/blob/d18036c259bbc47cbd4200f2428258934e935699/optics/Lens_Order.thy\#L89}
\end{definition}
\noindent We define $\lequiv$ as the cycle of a preorder, and consequently we can prove that it forms an equivalence relation.
\begin{corollary}[Lens Equivalence Relation]
  For any $\src$, $(\ast \ltto \src, \lequiv)$ forms a setoid, that is, $\lequiv$ is an equivalence relation on the set
  $\ast \ltto \src$ -- it is reflexive, symmetric, and transitive. \isalink{https://github.com/isabelle-utp/utp-main/blob/d18036c259bbc47cbd4200f2428258934e935699/optics/Lens_Order.thy\#L103}
\end{corollary}
\begin{proof}
  Reflexivity and transitivity follow by \cref{thm:lpreorder}, and symmetry follows by definition.
\end{proof}

\noindent Lens equivalence is a heterogeneously typed relation that is different from equality ($X = Y$), since it
requires only that the two state spaces are the same, whilst the view types of $X$ and $Y$ can be
different. Consequently, it can be used to compare lenses of different view types and show that two sets of lenses are
isomorphic. This makes this relation much more useful for evaluating observational equivalence between two lenses that
have apparently differing views, and yet characterise precisely the same region. For example, in general we cannot show
that $x + y + z = z + y + x$, since these constructions have different view types:
$\view_x \times \view_y \times \view_z$ and $\view_z \times \view_y \times \view_x$, respectively, and so the formula is
not even type correct. We can, however, show that $x + y + z \lequiv z + y + x$. This is reflected by the following set of
algebraic laws.

\begin{theorem} If $X$, $Y$, and $Z$ are all total lenses then they satisfy the following identities: \isalink{https://github.com/isabelle-utp/utp-main/blob/d18036c259bbc47cbd4200f2428258934e935699/optics/Lens_Order.thy\#L198}
  \begin{align*}
    X \lplus (Y \lplus Z) ~~\lequiv~~& (X \lplus Y) \lplus Z & \text{if } X \lindep Y, X \lindep Z, Y \lindep Z \\
    X \lplus \lzero ~~\lequiv~~& X \\
    X \lplus Y ~~\lequiv~~& Y \lplus X & \text{if } X \lindep Y
  \end{align*}
\end{theorem}

\noindent Lens summation is associative, has $\lzero$ as a unit, and commutative, modulo $\lequiv$, and assuming independence of
the components, which is consistent with the lens set interpretation. Lenses thus form a partial commutative
monoid~\cite{Dongol2015}, modulo $\lequiv$, also known as a separation algebra~\cite{Calcagno2007}, where $X \lplus Y$
is effectively defined only when $X \bowtie Y$. Independence corresponds with separation algebra's ``separateness''
relation, which means that there is no overlap between two areas of memory. We can also use $\lequiv$ to determine
whether two independent lenses, $X$ and $Y$, partition the entire state space using the identity $X +
Y \lequiv \lone$~\cite{Foster20-LocalVars}. Finally, we can prove the following additional properties of equivalence.

\begin{theorem}If $X_1$, $X_2$, $Y_1$, $Y_2$, and $Y$ are total lenses then the following laws hold: \isalink{https://github.com/isabelle-utp/utp-main/blob/1218a18a366a50b1cd49257523fbe1900adbe10b/optics/Lens_Order.thy\#L291}
  \begin{itemize}
  \item If $X_1 \lindep Y$ and $X_1 \lequiv X_2$ then $X_2 \lindep Y$;
  \item If $X_1 \lequiv X_2$, $Y_1 \lequiv Y_2$, and $X_1 \lindep Y_1$ then $X_1 + Y_1 \lequiv X_2 + Y_2$;
  \item If $\src_{X_1} = \view_Y$ and $X_1 \lequiv X_2$ then $X_1 \lcomp Y \lequiv X_2 \lcomp Y$.
  \end{itemize}
\end{theorem}

\noindent Independence is, as can be expected, preserved by equivalence. Equivalence is a congruence
relation with respect to $+$, provided the summed lenses are independent. It is also a left congruence for lens
composition. We can neither prove that it is a right congruence, nor find a counterexample.

\subsection{Mechanised State Spaces}
\label{sec:mechss}

Lenses allow us to express provisos in laws with side conditions about variables. Manual construction of state spaces
using the lens combinators is tedious and so we have implemented an Isabelle/HOL command for automatically creating a new
state space with the following form:
\isalink{https://github.com/isabelle-utp/utp-main/blob/d18036c259bbc47cbd4200f2428258934e935699/optics/Lens_Record.ML}
$$\ckey{alphabet}~ [\alpha_1, \cdots, \alpha_k]S = ([\beta_1, \cdots, \beta_m]T \,+)^? \,\, x_1 : \tau_1 ~~ \cdots ~~ x_n : \tau_n$$

\noindent We name the command $\ckey{alphabet}$, since it effectively allows the definition of a UTP alphabet (see
\cref{sec:UTP}), which in turn induces a state space. The command creates a new state space type $S$ with $k$ type
parameters ($\alpha_i$, for $1 \le i \le k$), optionally extending the parent state space $T$ with $m$ type parameters
($\beta_i$, for $1 \le i \le m$), and creates a lens for each of the variables $x_i : \tau_i$.  It can be used to
describe a concrete state space for a program. For brevity, we often abbreviate the \ckey{alphabet} command by the
syntax
$$\ssdef[\lbrack\beta_1, \cdots, \beta_m\rbrack T]{[\alpha_1, \cdots, \alpha_k]~S}{x_1 : \tau_1, ~~ \cdots ~~, x_n : \tau_n}$$ when
used in mathematical definitions. Internally, the command performs the following steps:

\begin{enumerate}
  \item generates a record space type $S$ with $n$ fields, which optionally extends a parent state space $T$;
  \item generates a lens $x_i$ for each of the fields using the record lens $\lrec^{R}_{x_i}$;
  \item automatically proves that each lens $x_i$ is a total lens;
  \item automatically proves an independence theorem $x_i \lindep x_j$ for each pair $i, j \in \{1 \cdots n\}$ such that $i \neq j$;
  \item generates lenses $\lbase_S$ and $\lmore_S$ that characterise the ``base part'' and ``extension part'', respectively;
  \item automatically proves a number of independence and equivalence properties.
\end{enumerate}

\noindent We now elaborate on each of these steps in detail.

The new record type $\textit{S}$ yields an auxiliary type $[\alpha_1, \cdots, \alpha_k, \phi]\textit{S-ext}$ with
additional type parameter $\phi$ that characterises future extensions. In particular, the non-extended type
$[\alpha_1, \cdots, \alpha_k]S$ is characterised by $[\alpha_1, \cdots, \alpha_k, \textit{unit}]\textit{S-ext}$, where
\textit{unit} is a distinguished singleton type. This extensible record type is isomorphic to a product of three basic
component types:
$$([\beta_1, \cdots, \beta_m]T \times (\tau_1 \times \cdots \times \tau_n)) \times \phi$$
These characterise, respectively, the part of state space described by $T$, the part described by the $n$ additional
fields, and the extension part $\phi$. In the case that the state space does not extend an existing type, we can set
$[\beta_1, \cdots, \beta_m]T = \textit{unit}$.

For each field, the command generates a lens $x_i : \tau_i \lto [\alpha_1, \cdots, \alpha_k, \phi]\textit{S-ext}$ using
the record lens, and proves total lens and independence theorems. Each of these lenses is polymorphic over $\phi$, so
that they can be applied to the base type and any extension thereof, in the style of inheritance in object oriented data
structures. As we show in \S\ref{sec:utp-thy}, this polymorphism allows us to characterise a hierarchy of UTP
theories.

In addition to the field lenses, we create two special total lenses:
\begin{itemize}
  \item $\lbase_S : [\alpha_1, \cdots, \alpha_k]S \lto [\alpha_1, \cdots, \alpha_k, \phi]\textit{S-ext}$, which characterises
the base part; and 
  \item $\lmore_S : \phi \lto [\alpha_1, \cdots, \alpha_k, \phi]\textit{S-ext}$, which characterises the extension part. 
  \end{itemize}
  The base part consists of only the inherited fields and those added by $S$. We automatically prove a
number of theorems about these special lenses:

\begin{itemize}
  \item $\lbase_S \lindep \lmore_S$: the base and extension parts are independent;
  \item $\lbase_S \lplus \lmore_S \lequiv \lone$: they partition the entire state space;
  \item for $i \in \{1 \cdots n\}$, $x_i \lsubseteq \lbase_S$: each variable lens is part of the base;
  \item $\lbase_S \lequiv \lbase_T \lplus \left(\sum_{i \in \{1 \cdots n\}} x_i\right)$: the base is composed of the
    parent's base and the variable lenses;
  \item $\lmore_T \lequiv \left(\sum_{i \in \{1 \cdots n\}} x_i\right) \lplus \lmore_S$: the parent's extension is composed
    of the variable lenses and the child's extension part.
\end{itemize}

\noindent These theorems can help support the Isabelle/UTP laws of programming, which we elaborate in the next
section. We emphasise, though, that the \ckey{alphabet} command is not the only way to construct a state space with
lenses, and nor do the results that follow depend on the use of this command. We could, for example, axiomatise a
collection of lenses, including independence relations over an abstract state space type, following Schirmer and
Wenzel~\cite{Schirmer2009}. However, the \ckey{alphabet} command is a convenient tool in many circumstances.

\section{Mechanising the UTP Relational Calculus}
\label{sec:isabelle-utp}

In this section we describe the core of Isabelle/UTP, including its expression model, meta-logical operators, predicate
calculus, and relational calculus, building upon our algebraic model of state spaces.  A direct result is an expressive
model of relational programs which can be used in proving fundamental algebraic laws of programming~\cite{Hoare87}, and
for formal verification (\S\ref{sec:verify-rel}).  Moreover, the relational model is foundational to the mechanisation
of UTP theories, and thus advanced computational paradigms, as we consider in \S\ref{sec:utp-thy}. An overview of the
Isabelle/UTP concepts and theories can be found in \cref{fig:UTP-Concepts} at the end of the paper.

\subsection{Expressions}
\label{sec:utp-expr}

Expressions are the basis of all other program and model objects in Isabelle/UTP, in that every such object is a
specialisation of the expression type. An expression language typically includes literals, variables, and function
symbols, all of which are also accounted for here. We model expressions as functions on the observation space:
$[A, \src]\uexpr \cong (\src \to A)$\footnote{We use a \ckey{typedef} to create an isomorphic but distinct type.  This
allows us to have greater control over definition of polymorphic constants and syntax translations, without
unnecessarily constraining these for the function type.}, where $A$ is the return type, which is a standard shallow
embedding approach~\cite{vonWright1994-RefHOL,Back1998,Feliachi2010}. However, lenses allow us to formulate syntax-like
constraints, but without the need for deeply embedded expressions. A major advantage of this model is that we need not
preconceive of all expression constructors, but can add them by definition.

We diverge from the standard shallow embedding approach, because we give explicit constructors for expressions. Usually
shallow embeddings use syntax translations to transparently map between program expressions and the equivalent lifted
expressions, for example, $$e + (f - g) \leadsto \lambda s @ e(s) + (f(s) - g(s))$$ Here, we prefer to have expression
constructors as first-class citizens.

\begin{definition}[Expression Constructors] \label{def:expconstr} Assume types $A$, $B$, $C$, and $\src$. We declare the constants: \isalink{https://github.com/isabelle-utp/utp-main/blob/d18036c259bbc47cbd4200f2428258934e935699/utp/utp_expr.thy}
$$\begin{array}{l}
  \uvare              ~\colon (A \lto \src) \to [A, \src]\uexpr \\
  \uvare~x            \defs \lambda s @ \lget_x~s \\[1ex]
  \ulit               ~\colon A \lto [A, \src]\uexpr \\
  \ulit~k     \defs \lambda s @ k \\[1ex]
  \uecond             ~\colon [\mathbb{B},\src]\uexpr \to [A, \src]\uexpr \to [A, \src]\uexpr \to [A, \src]\uexpr \\
  \uecond~b~u_1~u_2 \defs \lambda s @ \hifthenelse{b(s)}{u_1(s)}{u_2(s)} \\[1ex]
  \uuop               ~\colon (A\!\to\!B) \to [A, \src]\uexpr \to [B, \src]\uexpr \\
  \uuop~f~u   \defs \lambda s @ f(u(s)) \\[1ex]
  \ubop               ~\colon (A\!\to\!B\!\to\!C) \to [A, \src]\uexpr \to [B, \src]\uexpr \to [C, \src]\uexpr \\
  \ubop~g~u~v \defs \lambda s @ g(u(s))~(v(s))
\end{array}$$
where $x : A \lto \src$ is a lens; $k : A$ is a HOL constant; $f : A \to B$ and $g : A \to B \to C$ are functions; and
$b : [\mathbb{B},\src]\uexpr$, $u, u_1, u_2 : [A, \src]\uexpr$, and $v : [B, \src]\uexpr$ are expressions.
\end{definition}

\noindent The mechanisation of the core expression language is shown in \cref{fig:utp-expr}, which uses Isabelle's
lifting package~\cite{Huffman13} to create each of the expression constructors. The operator $\uvare~x$ is a variable
expression, and returns the present value in the state characterised by lens $x$. For convenience, we assume that $x$,
$y$, $z$, and decorations thereof, are lenses, and often use them directly as variable expressions without explicitly
using $\uvare$. We also use $\uv$ to denote the $\lone$ lens; this is effectively a special variable for the entire
state.

The operator $\ulit~k$ represents a literal, or alternatively an arbitrary lifted HOL value, and corresponds to a
constant function expression. We use the notation $\ulite{k}$ to denote a literal $k$. As well as lens-based variables,
which are used to model program variables, expressions can also contain HOL variables, which are orthogonal and constant
with respect to the program variables. HOL variables in literal constructions ($\ulite{\mv{x}}$) correspond to logical
variables~\cite{Mor1996}, also called ``ghost variables'', which are important for
verification~\cite{Hoare69,Mor1996}. We use the notations $\mv{x}$, $\mv{y}$, and $\mv{z}$ to denote logical
variables in expressions.

Operator $\uecond~b~u_1~u_2$ denotes a conditional expression; if $b$ is true then it returns $u_1$, otherwise $u_2$. It
evaluates the boolean expression $b$ under the incoming state, and chooses the expression based on this. We use the
notation $\econd{u_1}{b}{u_2}$ adopted in the UTP as a short hand for it.

Operators $\uuop$ and $\ubop$ lift HOL functions into the expression space by a pointwise lifting. With them we can
transparently use HOL functions as UTP expression functions, for instance the summation $e + f$ is denoted by
$\ubop~(+)~e~f$. Moreover, it is often possible to lift theorems from the underlying operators to the expressions
themselves, which allows us to reuse the large library of HOL algebraic structures in Isabelle/UTP. For instance, if we
know that $(A, +, 0)$ is a monoid, then also we can show that for any $\src$, $([A,\src]\uexpr, \ubop~(+), \ulit~0)$
forms a monoid. For convenience, we therefore often overload mathematically defined functions as expression constructs
without further comment. In particular, we often overload $=$ as both equivalence of two expressions ($e = f$), and an
expression of equality within an expression ($x = 5$).

\begin{figure}
  \begin{center}
  \includegraphics[width=.9\linewidth]{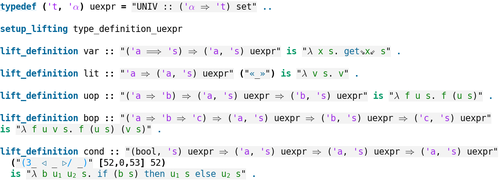}
\end{center}

  \vspace{-3ex}
  \caption{UTP Expression Model in Isabelle/HOL}
  \label{fig:utp-expr}

    \vspace{-2ex}
\end{figure}

This deep expression model allows us to mimick reasoning usually found in a deep embedding: the constructors above are
like datatype constructs, but are really semantic definitions. We can then prove theorems about these constructs that
allow us to reason in an inductive way, which is central to our approach to meta-logical reasoning. At the same time, we
have developed a lifting parser in Isabelle/UTP, which allows automatic translation between HOL expressions and UTP
expressions. We also have a tactic, \textsf{rel-auto}, that quickly and automatically eliminates the expression
structure, resulting in a HOL expression.

The \textsf{rel-auto} tactic performs best when $\src$ is constructed using the \ckey{alphabet} command of
\S\ref{sec:mechss}, because we can enumerate all the field lenses.  Given a state space
$[x_1 : \tau_1, \cdots x_n  \tau_n]$, we can eliminate the $s$ state space variable in a proof goal by replacing it
with a tuple of logical variables ($\mv{x}_1 \cdots \mv{x}_n$). This, in turn, means that we can eliminate each lens and
replace it with a corresponding logical variable, for example:
$$(x + y) - z > 3 \leadsto (\mv{x} + \mv{y}) - \mv{z} > 3$$ The result is a simpler expression containing only logical
variables, though of course with the loss of lens properties. This means that we have both the additional expressivity
and fidelity afforded by lenses, and the proof automation of Isabelle/HOL.

\subsection{Predicate Calculus}
\label{sec:predcalc}

A predicate is an expression with a Boolean return type, $[\src]\upred \defs [\mathbb{B}, \src]\uexpr$, so that
predicates are a subtype of expressions. The majority of predicate calculus operators ($\neg$, $\land$, $\lor$,
$\implies$) are obtained by pointwise lifting of the equivalent operators in HOL. We also define the indexed operators
$\bigwedge_{i \in A}~P(i)$ and $\bigvee_{i \in A}~P(i)$ similarly. The quantifiers are defined below. In order to
notationally distinguish HOL from UTP operators, in the following definitions we subscript them with a \textsf{H}.
\begin{definition}[Predicate Calculus Operators] \label{def:predop} \isalink{https://github.com/isabelle-utp/utp-main/blob/d18036c259bbc47cbd4200f2428258934e935699/utp/utp_pred.thy\#L289}
  \begin{align*}
    \exists x @ P &\defs \lambda s @ \holop{\exists} v : \view_x @ P(\lput_x~s~v) \\
    \bm{\exists} \mv{x} @ P(\mv{x}) &\defs \lambda s @ \holop{\exists} \mv{x} @ P(\mv{x})(s) \\
    [P] &\defs \forall \uv @ P \\
    (P \refinedby Q) &\defs (\holop{\forall} s @ (Q \implies P) s)
  \end{align*}
\end{definition}
\noindent Existential quantification over a lens $x$ quantifies possible values for the lens $v : \view_x$, and updates
the state with this using $\lput$. Universal quantification is obtained by duality. The emboldened existential
quantifier, $\bm{\exists}$, quantifies a logical variable in a parametric predicate $P(\mv{x})$ by a direct lifting of
the corresponding HOL quantifier. We emphasise that $\exists$ and $\bm{\exists}$ are semantically very different: in a
lens quantification, $\exists x @ P$, the lens $x$ must be an existing lens or expression. This lens is not bound by the
quantification, unlike $\bm{\exists} \mv{x} @ P(\mv{x})$ where $\mv{x}$ becomes a logical variable bound in $P$. Lens
quantification is elsewhere called \textit{liberation}~\cite{Dongol19,Colvin2017} since it removes any restrictions on
the valuation of $x$.

The universal closure, $[P]$, universally quantifies every variable in the
alphabet of $P$ using the state variable $\uv$. The refinement relation $P \refinedby Q$ is then defined as a HOL
predicate, requiring that $Q$ implies $P$ in every state $s$.

Since the definitions are by lifting of the underlying HOL operators, we obtain the following theorem.

\begin{theorem} \label{thm:predcba} For any $\src$, $([\src]\upred, \bigvee, \pfalse, \bigwedge, \ptrue, \neg)$ forms a complete Boolean algebra, that is: \isalink{https://github.com/isabelle-utp/utp-main/blob/d18036c259bbc47cbd4200f2428258934e935699/utp/utp_pred_laws.thy}
\begin{itemize}
  \item $([\src]\upred, \lor, \pfalse, \land, \ptrue, \neg)$ is a Boolean algebra, and
  \item $([\src]\upred, \refinedby)$ is a complete lattice with infimum $\bigvee$, supremum $\bigwedge$, top element
     $\pfalse$, and bottom $\ptrue$.
\end{itemize}
\end{theorem}

\noindent As usual, via the Knaster-Tarski theorem, for any monotonic function $F : [\src]\upred \to [\src]\upred$ we
can describe the least and greatest fixed points, $\mu F$ and $\nu F$, which in UTP are called the weakest and strongest
fixed points, and obey the usual fixed point laws. We can also algebraically characterise the UTP variable quantifiers
using Cylindric Algebra~\cite{Tarski71}, which axiomatises the quantifiers of first-order logic.
\begin{theorem} For any $\src$, $([\src]\upred, \lor, \land, \neg, \pfalse, \ptrue, \exists, =)$ forms a Cylindric
  Algebra, meaning that the following laws are satisfied for total lenses $x$, $y$, and $z$: \isalink{https://github.com/isabelle-utp/utp-main/blob/d18036c259bbc47cbd4200f2428258934e935699/utp/utp_pred_laws.thy\#L959}
  \begin{align*}
    (\exists x @ \pfalse) &= \pfalse \tag{C1} \label{law:C1} \\[.5ex]
    (\exists x @ P) &\refinedby P  \tag{C2} \label{law:C2} \\[.5ex]
    (\exists x @ (P \land (\exists x @ Q))) &= ((\exists x @ P) \land (\exists x @ Q)) \tag{C3} \label{law:C3} \\[.5ex]
    (\exists x @ \exists y @ P) &= (\exists y @ \exists x @ P) \tag{C4} \label{law:C4} \\[.5ex]
    (x = x) &= \ptrue \tag{C5} \label{law:C5} \\[.5ex]
    (y = z) &= (\exists x @ y = x \land x = z) & \textrm{if}~x \lindep y, x \lindep z \tag{C6} \label{law:C6} \\[.5ex]
    \pfalse &= \left(\begin{array}{l}
      (\exists x @ x = y \land P) \land \\
      (\exists x @ x = y \land \neg P)
    \end{array} \right) & \textrm{if}~x \lindep y \tag{C7} \label{law:C7}
  \end{align*}
\end{theorem}
\noindent From this algebra, the usual laws of quantification can be derived~\cite{Tarski71}. These laws illustrate the
difference in expressive power between HOL and UTP variables. For the former, we cannot pose meta-logical questions like
whether two variable names $x$ and $y$ refer to the same region, such as may be the case if they are aliased. For this
kind of property, we can use lens independence $x \lindep y$, as required by laws \ref{law:C6} and \ref{law:C7}. We
also prove the following laws for quantification.

\begin{theorem} If $A$ and $B$ are total lenses, then the following identities hold: \isalink{https://github.com/isabelle-utp/utp-main/blob/d18036c259bbc47cbd4200f2428258934e935699/utp/utp_pred_laws.thy\#L656}
\begin{align*}
  (\exists A \lplus B @ P) &= (\exists A @ \exists B @ P) & \textrm{if}~A \lindep B \tag{Ex1} \label{law:ExPlus} \\
  (\exists B @ \exists A @ P) &= (\exists A @ P) & \textrm{if}~B \lsubseteq A \tag{Ex2} \label{law:ExSubsume} \\
  (\exists A @ P) &= (\exists B @ P) & \textrm{if}~A \lequiv B \tag{Ex3} \label{law:ExEquiv}
\end{align*}
\end{theorem}
\noindent Here, lenses $A$ and $B$ can be interpreted as variable sets. \ref{law:ExPlus} shows that quantifying over two
disjoint sets of variables equates to quantification over both. Disjointness of variable sets is modelled by requiring
that the corresponding lenses are independent. \ref{law:ExSubsume} shows that quantification over a larger lens subsumes
a smaller lens. Finally \ref{law:ExEquiv} shows that if we quantify over two lenses that identify the same subregion
then those two quantifications are equal. We now have a complete set of operators and laws for the predicate calculus.

\subsection{Meta-Logic}
\label{sec:meta-logic}

Lenses treat variables as semantic objects that can be checked for independence, ordered, and composed in various
ways. As we have noted, we can consider such manipulations as meta-logical with respect to the predicate. We add further
specialised meta-logical queries for expressions.

Often we want to check which regions of the state space an expression depends on, for example to support laws of
programming and verification calculi, like Example~\ref{ex:asn-comm}. In a deep embedding, this is characterised
syntactically using free variables. However, lenses allow us to characterise a corresponding semantic notion called
``unrestriction'', which is originally due to Oliveira et al.~\cite{Oliveira07}.

\begin{definition}[Unrestriction] \label{def:unrest} \isalink{https://github.com/isabelle-utp/utp-main/blob/d18036c259bbc47cbd4200f2428258934e935699/utp/utp_unrest.thy\#L43}
  \begin{align*}
    \text{-} \unrest \text{-} ~~:~~~& (A \lto S) \to [B, S]\uexpr \to \mathbb{B} \\
    (x \unrest e) ~\defs~~&\left(\forall (s, k) @ e(\lput_x~s~k) = e(s)\right)
  \end{align*}
\end{definition}
\noindent Intuitively, lens $x$ is unrestricted in expression $e$, written $x \unrest e$, provided that $e$'s valuation
does not depend on $x$. Specifically, the effect of $e$ evaluated under state $s$ is the same if we change the value of
$x$. For example, $x\,\unrest\,(y + 2 > y)$ is true, provided that $x \lindep y$, since the truth value of $y + 2 > y$
is unaffected by changing $x$. Unrestriction is a weaker notion than free variables: if $x$ is not free in $e$ then
$x \unrest e$, but not the inverse. For example, $x \unrest\, (x = \ptrue \lor x = \pfalse)$ for $x : \bool$ is true,
since this expression is always true no matter the valuation of $x$. As we shall see, unrestriction is a sufficient
notion to formalise the provisos for the laws of programming. Below are some key laws for establishing whether an
expression is unrestricted by a variable.
\begin{lemma}[Unrestriction Laws] If $x$ and $y$ are total lenses, then the following laws hold:
 
  \begin{center}
  \begin{tabular}{ccccccc}
  \AxiomC{---\vphantom{$P$}}
  \UnaryInfC{$x \unrest \ulit~k$}
  \DisplayProof 
  &
  \AxiomC{$x \lindep y$}
  \UnaryInfC{$x \unrest \uvare~y$}
  \DisplayProof
  &
  \AxiomC{$x \unrest u$}
  \UnaryInfC{$x \unrest \uuop~f~u$}
  \DisplayProof
  &
  \AxiomC{$x \unrest u$}
  \AxiomC{$x \unrest v$}
  \BinaryInfC{$x \unrest \ubop~f~u~v$}
  \DisplayProof
  &
  \AxiomC{---\vphantom{$P$}}
  \UnaryInfC{$\lzero \unrest u$}
  \DisplayProof
  &
  \AxiomC{$x \lsubseteq y$}
  \AxiomC{$y \unrest u$}
  \BinaryInfC{$x \unrest u$}
  \DisplayProof
  &
  \AxiomC{$x \lindep y$}   
  \AxiomC{$x \unrest u$}
  \AxiomC{$y \unrest u$}
  \TrinaryInfC{$(x \lplus y) \unrest u$}
  \DisplayProof
  \end{tabular}

  \vspace{2ex}

  \begin{tabular}{ccccccc}
    \AxiomC{$x \lsubseteq y$}
    \UnaryInfC{$x \unrest (\exists y @ P)$}
    \DisplayProof
    &
    \AxiomC{$x \lindep y$}
    \AxiomC{$x \unrest P$}
    \BinaryInfC{$x \unrest (\exists y @ P)$}
    \DisplayProof
    &
    \AxiomC{$x \lsubseteq y$}
    \UnaryInfC{$x \unrest (\forall y @ P)$}
    \DisplayProof
    &
    \AxiomC{$x \lindep y$}
    \AxiomC{$x \unrest P$}
    \BinaryInfC{$x \unrest (\forall y @ P)$}
    \DisplayProof
  \end{tabular}
\end{center}

\end{lemma}

\noindent These laws are formulated in the style of an inductive definition, but in reality they are a set of lemmas in
Isabelle over our deep expression model. Expression $\ulit~k$ does not depend on the state since it always returns $k$:
any lens $x$ is unrestricted. Any lens $y$ that is independent of $x$ is unrestricted $\uvare~x$. The laws for $\uuop$
and $\ubop$ require simply that the component expressions have lens $x$ unrestricted. The $\lzero$ lens is unrestricted
in any expression $u$, since it characterises none of the state space.

Unrestriction is preserved by the lens order $\lsubseteq$. The summation of lenses $x$ and $y$ is unrestricted in $u$
provided that $x$ and $y$ are independent, and both are unrestricted in $u$. Again, we note that $\lplus$ can
effectively be used to group lenses in order to characterise a set of variables. Following a lens quantification over
$y$, any sublens $x$ of $y$ becomes unrestricted. On the other hand, the restriction of any lens independent of $x$
is unchanged. The final two laws are the dual case for the universal lens quantification.

We can also prove the following correspondence between predicate unrestriction and quantification.

\begin{lemma} \label{thm:ex-unrest} If $x$ is a total lens then $(x \unrest P) \iff (\exists x @ P) = P$. \isalink{https://github.com/isabelle-utp/utp-main/blob/d18036c259bbc47cbd4200f2428258934e935699/utp/utp_pred_laws.thy\#L724} \end{lemma}

\noindent We can alternatively characterise unrestriction of $x$ by showing that quantifying over $x$ in $e$ has no effect (it is a fixed point), which is a well-known property
from Cylindric Algebra~\cite{Tarski71} and also Liberation Algebra~\cite{Dongol19,Colvin2017}. Specifically, quantification
liberates the lens $x$ so that it is free to take any value. Nevertheless, \cref{thm:ex-unrest} cannot replace
Definition~\ref{def:unrest}, as it applies only if $P$ is a predicate, and not for an arbitrary expression. We can prove a
number of useful corollaries for quantification.

\begin{corollary} If $x \unrest P$ then $(\exists x @ P) = P$ and $(\exists x @ P \land Q) = (P \land (\exists x @ Q))$. \isalink{https://github.com/isabelle-utp/utp-main/blob/d18036c259bbc47cbd4200f2428258934e935699/utp/utp_pred_laws.thy\#L740}
\end{corollary}
\noindent The cases for universal quantification ($\forall$) also hold by duality.

Aside from checking for use of variables, UTP theories often require that the state space of a predicate can be extended
with additional variables. An example of this is the variable block operator that adds a new local variable. In
Isabelle/UTP, alphabets are implicitly characterised by state-space types, rather than explicitly as sets of
variables. Consequently, we perform alphabet extensions by manipulation of the underlying state space using lenses. We
define the following operator to extend a state space.

\begin{definition}[Alphabet Extrusion] \isalink{https://github.com/isabelle-utp/utp-main/blob/d18036c259bbc47cbd4200f2428258934e935699/utp/utp_alphabet.thy\#L36}
  \begin{align*}
    \text{-} \!\aext \text{-} ~~:~~~& [A, S_1]\uexpr \to (S_1 \lto S_2) \to [A, S_2]\uexpr \\
    e \aext a       ~~\defs~~& \left(\lambda s @ e(\lget_a~s)\right) 
  \end{align*}
\end{definition}

\noindent Alphabet extrusion, $e \aext a$, uses the lens $a : S_1 \lto S_2$ to alter the state space of $e : [A,
  S_1]\uexpr$ to become $S_2$. The lens effectively describes how one state space can be embedded into another. The
operator can be used either to extend a state space, or coerce it to an isomorphic one. For example, if we assume we
have a state space composed of two sub-regions: $\src \defs A \times B$, then, a predicate $P : [A]\upred$, which acts on
state space $A$, can be coerced to one acting on $\src$ by an alphabet extrusion: $P \oplus \lfst :
[\src]\upred$. Naturally, we know that the resulting predicate is unrestricted by the $B$ region.

\begin{lemma}[Alphabet Extrusion Laws] \isalink{https://github.com/isabelle-utp/utp-main/blob/d18036c259bbc47cbd4200f2428258934e935699/utp/utp_alphabet.thy\#L58}
  $$\begin{array}{rclrcl}
    \ulit~k     \aext a &=& \ulit~k  & \quad \uvare~x    \aext a &=& \uvare~(x \lcomp a) \\
    \uuop~f~u   \aext a &=& \uuop~f~(u \aext a) & \quad \ubop~g~u~v \aext a &=& \ubop~g~(u \aext a)~(v \aext a)
  \end{array}$$
\end{lemma}

\noindent Alphabet extrusion has no effect on a literal expression, other than to change its type, because it refers to no
lenses. A variable constructor has its lens augmented by composing it with the alphabet lens $a$. The extrusion simply
distributes through unary and binary expressions.

\subsection{Substitutions}
\label{sec:subst}

Substitution is an operator for replacing a variable in an expression with another expression: $e[f/x]$. Like free
variables, it is often considered as a meta-logical operator~\cite{Gordon1989-VerifyHOL}. However, shallow embeddings
can support a similar operator at the semantic level~\cite{Gordon1989-VerifyHOL,Back2005-Procedures}, for which the
substitution laws can be proved as theorems. Here, we generalise this using lenses, and treat substitutions as first-class
citizens that can contain multiple variable mappings, and also conditional substitutions. This allows us to unify
substitution and multiple variable assignment, which we demonstrate in \cref{thm:proglaws}~\eqref{law:LP4} and
\cref{thm:asnlaws}.

Substitutions can be modelled as total functions $\sigma, \rho : \src \to \src$ that transform an initial state to a
final state. However, it is more intuitive to consider a substitution as a set of mappings from variables to
expressions: $[x \smapsto e, y \smapsto f]$. The simplest substitution, $id \defs (\lambda x @ x)$, leaves the state
unchanged. We define operators for querying and updating our semantic substitution objects below.
\begin{definition}[Substitution Query and Update] \isalink{https://github.com/isabelle-utp/utp-main/blob/d18036c259bbc47cbd4200f2428258934e935699/utp/utp_subst.thy\#L72}
\begin{align*}
  \usublk{\sigma}{x} ~\defs~ (\lambda s @ \lget_x ~ (\sigma(s))) \qquad
  \sigma(x \smapsto e) ~\defs~ (\lambda s @ \lput_x ~ (\sigma(s)) ~ (e(s)))
\end{align*}
\end{definition}
\noindent Substitution query $\usublk{\sigma}{x} : [A, \src]\uexpr$ returns the expression associated with $x : A \lto
\src$ in $\sigma : \src \to \src$ by composition of the latter with $\lget_x$. Substitution update assigns the
expression $e : [A, \src]\uexpr$ to the lens $x$ in $\sigma$. The definition constructs a function of type
$\src \to \src$ that inputs the state $s$, evaluates $e$ with respect to $s$, calculates the state space updated by
$\sigma$, and then uses $\lput$ to update the value of $x$ in $\sigma(s)$ with the evaluated expression. We can then
introduce the short-hands
\begin{align*}
  \sigma(x_1 \smapsto e_1 , \cdots, x_n \smapsto e_n) &\defs (x_1 \smapsto e_1)\cdots(x_n \smapsto e_n) \\
  [x_1 \smapsto e_1 , \cdots, x_n \smapsto e_n] &\defs id(x_1 \smapsto e_1 , \cdots, x_n \smapsto e_n)
\end{align*}
which, respectively, update a
substitution $\sigma$ in $n$ variables by assigning $e_i$ to each variable $x_i$, and construct a new substitution from
a set of maplets. If we have $x_i \lindep x_j$ for each such variable, then these updates are effectively concurrent,
regardless of the expressions. Otherwise, a syntactically later assignment ($x_j$) can override an earlier one ($x_i$
with $i < j$). This allows us to model multiple variable assignment, and also evaluation contexts for programs and
expressions.

We also prove a number of laws about substitutions constructed from maplets.

\begin{lemma} \label{thm:subst} If $x$ and $y$ are total lenses, then the following identities hold: \isalink{https://github.com/isabelle-utp/utp-main/blob/bee1f650ce9a5f2f6faf7e5d76da99d4265cd5f0/utp/utp_subst.thy}
  \begin{align*}
    \usublk{\sigma(x \smapsto e)}{x} &= e \tag{SB1} \label{law:SB1} \\
    \sigma(x \smapsto \usublk{\sigma}{x}) &= \sigma \tag{SB2} \label{law:SB2} \\
    [x \smapsto x] &= id \tag{SB3} \label{law:SB3} \\
    \sigma(x \smapsto e, y \smapsto f) &= \sigma(y \smapsto f, x \smapsto e) & \textnormal{if } x \lindep y \tag{SB4} \label{law:SB4} \\
    \sigma(x \smapsto e, y \smapsto f) &= \sigma(y \smapsto f) & \textnormal{if } x \lsubseteq y \tag{SB5} \label{law:SB5}
  \end{align*}
\end{lemma}

\noindent Law~\ref{law:SB1} states that looking up $x$ in a constructed substitution returns its associated expression,
$e$. It essentially follows from lens axiom \ref{law:put-get}. Law~\ref{law:SB2} is an $\eta$-conversion principle,
which follows from \ref{law:get-put}. Law~\ref{law:SB3} states that updating a variable to itself has no
effect. Law~\ref{law:SB4} states that two substitution maplets, $x \smapsto e$ and $y \smapsto f$, commute provided that
$x$ and $y$ are independent. Similarly, \ref{law:SB5} states that a later assignment for $y$ overrides an earlier
assignment for $x$ when $y$ is a wider lens than $x$. This is, of course, true in particular when $x = y$, which reduces
to lens axiom \ref{law:put-put}.

Substitutions can be composed in sequence using function composition. Conditional substitutions can be expressed using
the following construct.

\begin{definition}[Conditional Substitution] $\scond{\sigma}{b}{\rho} ~\defs~ (\lambda s @ \hifthenelse{b(s)}{\sigma(s)}{\rho(s)})$. \isalink{https://github.com/isabelle-utp/utp-main/blob/d18036c259bbc47cbd4200f2428258934e935699/utp/utp_expr.thy\#L92}
\end{definition}

\noindent A conditional substitution is equivalent to $\sigma$ when $b$ is true, and $\rho$ otherwise. The definition
evaluates $b$ under the incoming state $s$, and then chooses which substitution to apply based on this. Substitutions
can be applied to an expression using the following operator.

\begin{definition}[Substitution Application] \isalink{https://github.com/isabelle-utp/utp-main/blob/d18036c259bbc47cbd4200f2428258934e935699/utp/utp_subst.thy\#L35}
  \begin{align*}
    -\!\usubapp\!- &~\colon (\src \to \src) \to [A, \src]\uexpr \to [A, \src]\uexpr \\
    \sigma \usubapp e &\triangleq (\lambda s @ e(\sigma(s)))
  \end{align*}
\end{definition}
\noindent Application of a substitution $\sigma$ to an expression $e$ simply evaluates $e$ in the context of state
$\sigma(s)$. We can also model the classical syntax for substitution, $P[v/x] \defs [x \smapsto v] \usubapp P$, and
prove the substitution laws~\cite{Back2005-Procedures}.
\begin{lemma}[Substitution Application Laws] \label{thm:substapp} \isalink{https://github.com/isabelle-utp/utp-main/blob/bee1f650ce9a5f2f6faf7e5d76da99d4265cd5f0/utp/utp_subst.thy}
  \begin{align*}
    \sigma \usubapp \uvare~x &= \usublk{\sigma}{x} \tag{SA1} \label{law:SA1} \\
    \usubupd{\sigma}{x}{e} \usubapp u &= \sigma \usubapp u & \textnormal{if } x \unrest u \tag{SA2}\label{law:SA2} \\
    \sigma \usubapp \uuop~f~v &= \uuop~f~(\sigma \usubapp v) \tag{SA3} \label{law:SA3} \\
    \sigma \usubapp \ubop~f~u~v &= \ubop~f~(\sigma \usubapp u)~(\sigma \usubapp v) \tag{SA4} \label{law:SA4} \\
    (\exists y @ P)[e/x] & = (\exists y @ P[e/x]) &\textrm{if}~x \lindep y, y \unrest e \tag{SA5}\label{law:SA5} \\
    id \usubapp e &= e \tag{SA6} \label{law:SA6} \\
    \sigma \usubapp \rho \usubapp e &= (\rho \circ \sigma) \usubapp e \tag{SA7} \label{law:SA7} \\
    \rho(x \smapsto e) \circ \sigma &= (\rho \circ \sigma)(x \smapsto \sigma \usubapp e) \tag{SA8} \label{law:SA8} \\
    \scond{\sigma(x \smapsto e)}{b}{\rho(x \smapsto f)} &= (\scond{\sigma}{b}{\rho})(x \smapsto (\econd{e}{b}{f})) \tag{SA9} \label{law:SA9}
  \end{align*}
\end{lemma}
\noindent Application of $\sigma$ to a variable $x$ is the valuation of $x$ in $\sigma$ (\ref{law:SA1}). A substitution maplet for
an unrestricted variable can be removed (\ref{law:SA2}). Substitutions distribute through both unary (\ref{law:SA3}) and
binary (\ref{law:SA4}) operators. A singleton substitution for variable $x$ can pass through an existential
quantification over $y$ provided that $x$ and $y$ are independent, and $e$ is unrestricted by $y$ (\ref{law:SA5}), which
prevents variable capture. Application of $id$ has no effect (\ref{law:SA6}), and application of two substitutions can
be expressed by their composition (\ref{law:SA7}). \ref{law:SA8} shows that when $\sigma$ is composed with another
substitution composed of maplets ($x \smapsto e$), it is simply applied to the expression $e$ of every such
maplet. \ref{law:SA9} shows how two substitutions with matching maplets can be conditionally composed, by distributing
the conditional.

We can also use substitutition and unrestriction to prove the one-point law~\cite{Hehner1990} from predicate calculus,
as employed in Hehner's classic textbook on predicative programming~\cite{Hehner93}.

\begin{theorem} If $x$ is a total lens and $x \unrest e$ then $(\exists x @ P \land x = e) = P[e/x]$ \isalink{https://github.com/isabelle-utp/utp-main/blob/d18036c259bbc47cbd4200f2428258934e935699/utp/utp_pred_laws.thy\#L658}
\end{theorem}

\noindent As for expressions, we define an operator to extend the alphabet of a substitution.

\begin{definition}[Substitution Alphabet Extrusion] \label{def:saext} $\sigma \saext a ~\defs~ (\lambda s @ \lput_a ~ s ~ (\sigma(\lget_a ~ s)))$ \isalink{https://github.com/isabelle-utp/utp-main/blob/d18036c259bbc47cbd4200f2428258934e935699/utp/utp_alphabet.thy\#L330}
\end{definition}

\noindent We use the lens $a : S_1 \lto S_2$ to coerce $\sigma : S_1 \to S_1$ to the state space $S_2$. The resulting substitution
first obtains an element of $S_1$ from the incoming state $s : S_2$ using $\lget_a$, applies the substitution to this,
and then places the updated state back into $s$ using $\lput_a$. This is the essence of a framed computation; the parts
of $s$ outside of the view of $a$ are unchanged.

\subsection{Relational Programs}
\label{sec:relcalc}

A relation is a predicate on a product space $\src_1 \times \src_2$, specifically, where $\src_1$ and $\src_2$ are the
state spaces before and after execution, respectively. All laws that have been proved for expressions and predicates
therefore hold for relations. We define types for both heterogeneous relations,
$[\src_1,\src_2]\urel \defs [\src_1 \times \src_2]\upred$, and homogeneous relations
$[\src]\hrel \defs [\src,\src]\urel$. Operators $\ptrue$ and $\pfalse$ can be specialised in this relational setting,
and stand for the most and least non-deterministic relations. Due to their special role, we use the notation
$\true$ and $\false$ to explicitly refer to these relational counterparts.

In common with formal languages like Z~\cite{Spivey89} and B~\cite{Abrial1996a}, UTP~\cite{Hoare&98} uses the notational
convention for variables that $x$ is the initial value and $x'$ is its final value. The former can be denoted by the
lens composition $x \lcomp \lfst$, and the latter by $x \lcomp \lsnd$, where $x$ is actually the name of a lens of type
$\tau \lto \src$. In this presentation we define the operators below for lifting variables, expressions, and
substitutions into the product space.

\begin{definition}[Pre- and Postcondition Lifting] \isalink{https://github.com/isabelle-utp/utp-main/blob/d18036c259bbc47cbd4200f2428258934e935699/utp/utp_lift.thy}
  $$\upre{x} \defs x \lcomp \lfst \qquad \upost{x} \defs x \lcomp \lsnd \qquad \upre{e} \defs e \aext \lfst \qquad \upost{e} \defs e \aext \lsnd \qquad \upre{\sigma} \defs \sigma \saext \lfst \qquad \upost{\sigma} \defs \sigma \saext \lsnd$$
\end{definition}

\noindent The $\upre{{}}$ and $\upost{{}}$ lift a lens, expression, or substitution, into the first and second
components of a product state space $\src_1 \times \src_2$. We deviate from the standard notation to avoid the ambiguity
between ``$x$'' meaning an expression variable or an initial relational variable. Operator $\upre{e}$ lifts an
expression $e : [A, \src_1]\uexpr$ to an expression on the product state space $\src_1 \times \src_2$, for any given
$\src_2$. If $e$ is a predicate on the state variables, then $\upre{e}$ is a predicate on the initial state, that is a
precondition. Similarly, $\upost{e}$ constructs a postcondition with state space $\src_1 \times \src_2$ from
$e : [A, \src_2]\uexpr$. The analogous operators $\upre{\sigma}$ and $\upost{\sigma}$ lift a substitution to the product
space.

We can now define the main programming operators of the relational calculus.
\begin{definition}[Programming Operators] \label{def:progop} $ $ \isalink{https://github.com/isabelle-utp/utp-main/blob/master/utp/utp_rel.thy} \newline %

  \vspace{-3ex}

  \noindent\parbox{.5\linewidth}{
    \begin{align*}
      (P \relsemi Q)        &~\defs~ (\bm{\exists} \mv{v}_0 @ P[\mv{v}_0/\upost{\uv}] \land Q[\mv{v}_0/\upre{\uv}]) \\
      \II                   &~\defs~ (\upost{\uv} = \upre{\uv}) \\
      \uassigns{\sigma}     &~\defs~ (\upost{\uv} = \sigma(\upre{\uv}))
    \end{align*}}
  \parbox{.5\linewidth}{
    \begin{align*}
      \conditional{P}{b}{Q} &~\defs~ (\upre{b} \land P) \lor (\neg \upre{b} \land Q) \\
      \while{b}{Q}          &~\defs~ \nu X @ \conditional{P \relsemi X}{b}{\II}
    \end{align*}}

  \vspace{-2ex}

\end{definition}

\noindent These broadly follow the common definitions given in relational calculus and the UTP book~\cite{Hoare&98}, but
with subtle differences due to our use of lenses. Relational composition $P \relsemi Q$ existentiality quantifies a
logical variable $\mv{v}_0$ that stands for the intermediate state between $P$ and $Q$. It is substituted as the final
state of $P$ (using $\upost{\uv}$), and the initial state of $Q$; the resulting predicates are then conjoined. This
definition yields a homogeneous composition operator, though in Isabelle/UTP composition is heterogeneous and has type
$[\src_1, \src_2]\urel \to [\src_2, \src_3]\urel \to [\src_1, \src_3]\urel$\footnote{Technically, in Isabelle/UTP we
  obtain $\relsemi$ by lifting the HOL relational composition operator, and so the equation in \cref{def:progop} is
  actually a special case theorem. Nevertheless, it gives adequate intuition for this paper, and avoids the introduction
  of a further layer of abstraction.}.

Relational identity ($\II$), or skip, equates the initial state with the final state, leaving all variables
unchanged. $\uassigns{\sigma}$ is a generalised assignment operator, originally proposed by Back and von
Wright~\cite[page~2]{Back1998} using a substitution $\sigma : \src \to \src$. Its definition states that the final state
is equal to the initial state with $\sigma$ applied. The substitution $\sigma$ can be constructed as a set of maplets,
so that a singleton assignment $x := v$ can be expressed as $\uassigns{x \smapsto v}$, and a multiple variable
assignment as $\uassigns{x_1 \smapsto v_1, \cdots, x_n \smapsto v_n}$. Since $x$ can be any lens in an assignment, we
can use it to assign any part of the state hierarchy that can be so characterised, for example an element of the heap
(\cref{ex:st-hp}) or the attribute of an object (\cref{ex:objlens}).

Conditional $\conditional{P}{b}{Q}$ states that if $b$ is true then behave like $P$, otherwise behave like $Q$. In the
UTP book~\cite{Hoare&98} the fact that $b$ acts on initial variables is a syntactic convention, whereas here this
wellformedness condition is imposed by construction using $\upre{b}$. For completeness, we use this operator, and the
fact that all predicates, including relations, form a complete lattice, to define the while loop operator
$\while{b}{P}$. We use the strongest fixed-point, $\nu$, since we use it for partial correctness verification in
\S\ref{sec:verify-rel}.

We now have all the operators of a simple imperative programming language, and can prove the laws of
programming~\cite{Hoare87,Hoare&98}, but with a generalised presentation.

\begin{theorem}[Generalised Laws of Programming] \label{thm:proglaws} $ $ \isalink{https://github.com/isabelle-utp/utp-main/blob/d18036c259bbc47cbd4200f2428258934e935699/utp/utp_rel_laws.thy} \newline %
  \vspace{-3ex}

  \noindent
  \parbox{.45\linewidth}{
  \begin{align*}
    (P \relsemi Q) \relsemi R &= P \relsemi (Q \relsemi R) \tag{LP1} \\
    \II \relsemi P &= P \relsemi \II = P \tag{LP2} \\
    \false \relsemi P &= P \relsemi \false = \false \tag{LP3} \\
    \uassigns{\sigma} \relsemi P &= \upre{\sigma} \usubapp P \tag{LP4} \label{law:LP4} \\
    \conditional{\uassigns{\sigma}}{b}{\uassigns{\rho}} &= \uassigns{\scond{\sigma}{b}{\rho}} \tag{LP5} \label{law:LP5}
  \end{align*}
  }
  \parbox{.55\linewidth}{
  \begin{align*}
    \conditional{(P \relsemi (\while{b}{P}))}{b}{\II}\, &= \while{b}{P} \tag{LP6} \\
    \conditional{P}{\ptrue}{Q} &= P \tag{LP7} \\
    \conditional{P}{\pfalse}{Q} &= Q \tag{LP8} \\
    \conditional{P}{\neg b}{Q} &= \conditional{Q}{b}{P} \tag{LP9} \\
    (\conditional{P}{b}{Q}) \relsemi R &= \conditional{(P \relsemi R)}{b}{(P \relsemi R)} \tag{LP10} \label{law:LP10}
  \end{align*}}

  \vspace{-2ex}
\end{theorem}

\noindent The majority of these are standard, and therefore we select only a few for commentary. \ref{law:LP4} is a
generalisation of the forward assignment law: an assignment by $\sigma$ followed by relation $P$ is equivalent to
$\sigma$ applied to the initial variables of $P$. The more traditional formulation~\cite{Hoare87,Hoare&98},
$x := e \relsemi P = P[e/x]$, is an instance of this law. Similarly, \ref{law:LP5} is a generalised conditional
assignment law which combines $\sigma$ and $\rho$ into a single conditional substitution. All the other assignment laws
can be proved, but we need some additional properties for relational substitutions, which are shown
below. \ref{law:LP10} would normally require as a proviso that $b$ is an expression in initial variables only, but in
our setting this fact follows by construction.

\begin{lemma}[Relational Substitutions and Assignment] \label{thm:relsubst} $ $ \isalink{https://github.com/isabelle-utp/utp-main/blob/d18036c259bbc47cbd4200f2428258934e935699/utp/utp_rel.thy\#L555} \newline
  \vspace{-4ex}

  \noindent
  \parbox{.35\linewidth}{
  \begin{align*}
    \upre{\sigma} \usubapp (P \relsemi Q) &= (\upre{\sigma} \usubapp P) \relsemi Q \tag{RS1} \label{law:RS1} \\
    \upost{\sigma} \usubapp (P \relsemi Q) &= P \relsemi (\upost{\sigma} \usubapp Q) \tag{RS2} \label{law:RS2}
  \end{align*}}
  \parbox{.65\linewidth}{
  \begin{align*}
    \upre{\sigma} \usubapp  \uassigns{\rho} &= \uassigns{\rho \circ \sigma} \tag{RS3} \label{law:RS3} \\
    \upre{\sigma} \usubapp (\conditional{P}{b}{Q}) &= \conditional{(\upre{\sigma} \usubapp P)}{\,(\sigma \usubapp b)\,}{(\upre{\sigma} \usubapp Q)} \tag{RS4} \label{law:RS4}
  \end{align*}}
  
  \vspace{-2ex}

\end{lemma}

\noindent \ref{law:RS1} shows that a precondition substitution applies only to the first element of a sequential composition, and
\ref{law:RS2} is its dual. \ref{law:RS3} shows that precondition substitution applied to an assignment can be expressed
as a composite assignment. \ref{law:RS4} shows that $\upre{\sigma}$ applied to a conditional distributes through all
three arguments. The precondition annotation is removed when applied to $b$ since this is a precondition expression
already. Combining \cref{thm:proglaws} with \cref{thm:relsubst} we can prove the following corollaries of generalised
assignment.

\begin{corollary} \label{thm:gen-asn-laws} $\uassigns{id} = \II \qquad \uassigns{\sigma} \relsemi \uassigns{\rho} = \uassigns{\rho \circ \sigma} \qquad \uassigns{\sigma} \relsemi (\conditional{P}{b}{Q}) = \conditional{(\uassigns{\sigma} \relsemi P)}{\sigma \usubapp b}{(\uassigns{\sigma} \relsemi Q)}$
\end{corollary}

\noindent Moreover, with \cref{thm:substapp}, \cref{thm:proglaws}, and \cref{thm:relsubst} we can prove
the classical assignment laws~\cite{Hoare87}.

\begin{corollary}[Assignment Laws] \label{thm:asnlaws} \isalink{https://github.com/isabelle-utp/utp-main/blob/d18036c259bbc47cbd4200f2428258934e935699/utp/utp_rel_laws.thy\#L322}
 $ $ \newline
  \vspace{-4ex}

  \noindent
  \parbox{.4\linewidth}{
  \begin{align*}    
  x := e \relsemi y := f &~=~ x, y := e, f[e/x] \\
  x := x &~=~ \II \\
  x, y := e, y &~=~ x := e
  \end{align*}}
  \parbox{.6\linewidth}{
  \begin{align*}
  x, y, z := e, f, g &~=~ y, x, z := f, e, g \\
    x := e \relsemi (\conditional{P}{b}{Q}) &~=~ \conditional{(x := e \relsemi P)}{b[e/x]}{(x := e \relsemi Q)} \\
  \conditional{x := e}{b}{x := f} &~=~ x := (\econd{e}{b}{f}) 
  \end{align*}}

  \vspace{-3ex}
\end{corollary}

\noindent With these laws we can collapse any sequential and conditional composition of assignments into a single
assignment~\cite{Hoare87}. We illustrate this with the calculation below.

\begin{example}[Assignment Calculation] Assume $x$, $y$, and $z$ are independent total lenses, then:
  \begin{align*}
    x := 3 \relsemi y := x^2 + 4 \relsemi z := z \cdot y + x
    &~=~ \uassigns{x \smapsto 3} \relsemi \uassigns{y \smapsto x^2 + 4} \relsemi \uassigns{z \smapsto z \cdot y + x} & \text{[Definition]} \\
    &~=~ \uassigns{[y \smapsto x^2 + 4] \circ [x \smapsto 3]} \relsemi \uassigns{z \smapsto z \cdot y + x} & \text{[\ref{thm:gen-asn-laws}]} \\
    &~=~ \uassigns{(id \circ [x \smapsto 3])(y \smapsto [x \smapsto 3] \usubapp x^2 + 4} \relsemi \uassigns{z \smapsto z \cdot y + x} & \text{[\ref{law:SA8}]} \\
    &~=~ \uassigns{x \smapsto 3, y \smapsto 3^2 + 4} \relsemi \uassigns{z \smapsto z \cdot y + x} & \text{[\ref{thm:substapp}]} \\
    &~=~ \uassigns{[z \smapsto z \cdot y + x] \circ [x \smapsto 3, y \smapsto 9 + 4]} & \text{[\ref{thm:gen-asn-laws}]} \\
    &~=~ \uassigns{x \smapsto 3, y \smapsto 13, z \mapsto z \cdot 13 + 3} & \qed
  \end{align*}
\end{example}

\noindent The result is a simultaneous assignment to the three variables, $x$, $y$, and $z$. The value of $z$ depends on
its initial value, which is unknown, and so the variable is retained.

We have shown in this section how lenses and our mechanisation of UTP support a relational program model that satisfies the
laws of programming. In the next section we use them to derive operational semantics and verification calculi in
Isabelle/UTP.

\section{Automating Verification Calculi}
\label{sec:verify-rel}
In this section we demonstrate how the foundations established in \S\ref{sec:lenses} and \S\ref{sec:isabelle-utp} can be
applied to automated program analysis. We show how concrete programs can be modelled, symbolically executed, and
automatically verified using mechanically validated operational and axiomatic semantics, including Hoare logic and
refinement calculus. Though the mechanisation of such calculi has been achieved several times
before~\cite{Gordon1989-VerifyHOL,vonWright1994-RefHOL,Nipkow2014-ConcreteSemantics,Armstrong2015}, the novelty of our
work is that lenses allow us to (1) express the meta-logical provisos of variables for the various calculi, (2) unify a
variety of state space models, and (3) reason about aliasing of variables using independence and containment. We
illustrate this with a number of standard examples, and a larger verification of a find-and-replace algorithm that
utilises lenses in modelling arrays.

We consider the encoding of programs in Isabelle/UTP in \S\ref{sec:encprog}, symbolic execution in \S\ref{sec:symexec},
Hoare logic verification in \S\ref{sec:hoarelogic}, and finally refinement calculus in \S\ref{sec:ref-calc}.

\subsection{Encoding Programs}
\label{sec:encprog}

Imperative programs can be encoded using the operators given in \S\ref{sec:isabelle-utp}, and a concrete state space with the
required variables. We can describe the factorial algorithm as shown in the following example.

\begin{example}[Factorial Program] \label{def:factorial} We define a state space, $\ssdef{sfact}{x : \nat, y : \nat}$,
  consisting of two lenses $x$ and $y$. We can then define a program for computing factorials as follows: \isalink{https://github.com/isabelle-utp/utp-main/blob/c3e0b2e6fdf2e07985e9b1c887c2fdc44f7f208b/utp/examples/factorial.thy\#L9}
  \begin{align*}
    pfact &\,: \nat \to [sfact]\hrel \\
    pfact(\mv{X}) &\defs \left(\begin{array}{l}
                                 x := \mv{X} \relsemi y := 1 \relsemi \\
                                 \ckey{while}~ x > 1 ~\ckey{do} \\
                                 \quad y := y * x \relsemi x := x - 1 \\
                                 \ckey{od}
                               \end{array}\right)
  \end{align*}
  Constant \textit{pfact} is a function taking a natural number $\mv{X}$ as input and producing a program of type
  $[sfact]\hrel$ that computes the factorial. The given value is assigned to UTP variable $x$, and $1$ is assigned to
  $y$. The program iteratively multiplies $y$ by $x$, and decrements $x$. In the final state, $y$ has the
  factorial. \qed
\end{example}

This program can be entered into Isabelle/UTP using almost the same syntax as shown above. This is illustrated in
Figure~\ref{fig:pfact}, where we use the \ckey{alphabet} command to create the state space, and then define the
algorithm over this space using the \ckey{definition} command. Our parser automatically distinguishes whether a variable name
is a lens, such as $x$, or logical variable, such as $\mv{X}$, using type inference.

Although in this case we have manually constructed a state space, we can also pass the program variables as parameters
to the program. Specifically, we can instead define
$$pfact(\mv{X} : \nat, x : \nat \lto \src, y : \nat \lto \src) \defs \cdots$$ which takes $x$ and $y$ as parameters,
mimicking the pass-by-reference mechanism, and hence have a program of type $[\src]\hrel$ that is polymorphic in the
state space $\src$. This is possible because lenses model variables as first-class citizens with the specific
implementation abstracted by the axioms. This is in contrast to other mechanisations where variables either must have
explicit names~\cite{Gordon1989-VerifyHOL,vonWright1994-RefHOL,Oliveira07} or are unstructured
entities~\cite{Feliachi2010,Armstrong2015}.

As usual, the algorithm includes an invariant annotation~\cite{Armstrong2015} that will be explained shortly. We
use this as a running example for the Isabelle/UTP symbolic execution and verification components.

\begin{figure}
  \begin{center}
    \includegraphics[width=.45\linewidth]{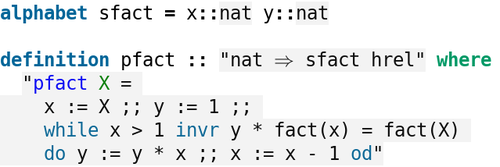}
  \end{center}

  \vspace{-3ex}

  \caption{Factorial Program in Isabelle/UTP}
  \label{fig:pfact}

    \vspace{-2ex}
\end{figure}

\subsection{Operational Semantics}
\label{sec:symexec}

Operational semantics are commonly assigned to a programming language by means of an inductive set of transition
rules~\cite{Nipkow2014-ConcreteSemantics}. UTP~\cite[chapter~10]{Hoare&98}, however, takes a different approach. As
usual it characterises a small-step operational semantics for imperative programs using a step relation,
$(s, P) \to (t, Q)$, meaning that program $P$, when started in state $s$, can transition to program $Q$, with state
$t$. However, instead of characterising this relation using an inductive set, it gives a denotational semantics to the
transition relation, and then proves transition rules as theorems. This has the benefit that linking operational and
axiomatic semantics is straightforward.

Here, we adopt a similar idea to describe a reduction semantics for relational programs, and use it to perform symbolic
execution, following the pattern given by Gordon and Collavizza~\cite{GordonC2010}. We begin by describing the execution
of a program $P$ started in a state context $\Gamma$.

\begin{definition}[Program Evaluation] $\opsem{\Gamma}{P} ~\defs~ \uassigns{\Gamma} \relsemi P$ where $P : [\src]\hrel$
  and $\Gamma : \src \to \src$. \isalink{https://github.com/isabelle-utp/utp-main/blob/c3e0b2e6fdf2e07985e9b1c887c2fdc44f7f208b/utp/utp_sym_eval.thy\#L10}
\end{definition}

\noindent The meaning of $\opsem{\Gamma}{P}$ is that program $P$ is executed in the variable context $\Gamma$, which
gives assignments to a selection of variables in $\src$, for example $[x \smapsto 5, y \smapsto \ptrue]$. The expression
$\opsem{\Gamma}{\II}$ denotes a program that has terminated in state $\Gamma$. The definition above of this operator is
very simple, because of our encoding of variable contexts as substitutions: we apply $\Gamma$ as an assignment, and
compose this with the program $P$. Using this definition we can prove a set of reduction theorems.

\begin{theorem}[Operational Reduction Rules] \label{thm:opsem} $ $ \isalink{https://github.com/isabelle-utp/utp-main/blob/c3e0b2e6fdf2e07985e9b1c887c2fdc44f7f208b/utp/utp_sym_eval.thy\#L15} \newline %

  \begin{center}
  \begin{tabular}{ccccccc}
  \AxiomC{---\vphantom{$P$}}
  \UnaryInfC{$\opsem{\Gamma}{P \relsemi Q} \rightarrow \opsem{\Gamma}{P} \relsemi Q$}
  \DisplayProof
  &
  \AxiomC{---\vphantom{$P$}}
  \UnaryInfC{$\opsem{\Gamma}{\II\,} \relsemi P \rightarrow \opsem{\Gamma}{P}$}
  \DisplayProof
  &
  \AxiomC{---\vphantom{$P$}}
  \UnaryInfC{$\opsem{\Gamma}{\uassigns{\sigma}} \rightarrow \opsem{\sigma \circ \Gamma}{\II\,}$}
  \DisplayProof \\[4ex]
  \end{tabular}
  \begin{tabular}{ccccccc}
  \AxiomC{$(\Gamma \usubapp b) \rightarrow \ptrue$}    
  \UnaryInfC{$\opsem{\Gamma}{\conditional{P}{b}{Q}} \rightarrow \opsem{\Gamma}{P}$}
  \DisplayProof
  &
  \AxiomC{$(\Gamma \usubapp b) \rightarrow \pfalse$}    
  \UnaryInfC{$\opsem{\Gamma}{\conditional{P}{b}{Q}} \rightarrow \opsem{\Gamma}{Q}$}
  \DisplayProof
  \\[4ex]
  \AxiomC{$(\Gamma \usubapp b) \rightarrow \ptrue$}    
  \UnaryInfC{$\opsem{\Gamma}{\while{b}{P}} \rightarrow \opsem{\Gamma}{P \relsemi (\while{b}{P})}$}
  \DisplayProof
  &
  \AxiomC{$(\Gamma \usubapp b) \rightarrow \pfalse$}    
  \UnaryInfC{$\opsem{\Gamma}{\while{b}{P}} \rightarrow \opsem{\Gamma}{\II\,}$}
  \DisplayProof
  \end{tabular}
  \end{center}  

\end{theorem}

\noindent The arrow $P \rightarrow Q$ actually denotes an equality predicate ($P = Q$); we use the arrow notation to aid
in comprehension, and to emphasise the left-to-right nature of semantic evaluations. The first two rules on the top line
handle sequential composition. The first pushes an evaluation into the first argument of a composition $P \relsemi
Q$. The second rule states that when the first argument of a sequential composition has terminated ($\II$), execution
moves on to the second argument ($P$). The third rule handles assignments, creating a new context by precomposing the
assignment $\sigma$ with the current context.

The second and third lines deal with conditional and while-loop iteration, respectively. In all these rules, the context
$\Gamma$ is applied to the condition $b$ as a substitution. If the result is $\ptrue$, the conditional chooses the first
branch $P$, and the while loop makes a copy of the loop body. If the result is $\pfalse$, the conditional chooses the
second branch $Q$, and the while loop terminates.

Using these theorems, we can perform symbolic execution of programs in Isabelle/UTP. To achieve this we load the
Isabelle simplifier with the equations of both \cref{thm:opsem} and \cref{thm:substapp}, the latter of which
allows us to apply substitutions. We also utilise the pointwise-lifted semantics given in \S\ref{sec:utp-expr} to
evaluate expressions using the builtin HOL functional definitions and simplification laws.

In addition, we need the equations of \cref{thm:subst} to evaluate and normalise substitutions. However,
Law~\ref{law:SB4}, which allows reordering of subsitution maplets, is symmetric and therefore cannot be directly used as
it would cause the simplifier to loop. The issue is that lenses do not \emph{apriori} have a total order that can be used
to reorder them. Nevertheless, Law~\ref{law:SB4} is important to enable a canonical representation of concrete
substitutions. Consequently, we extend the simplifier with a ``simproc''~\cite{Isabelle}, a specialised meta-logical
simplification procedure that sorts substitution maplets lexicographically using the syntactic lens names. Thus, during
symbolic evaluation the variable context will always order the variables lexicographically.

Applying the Isabelle simplifier with these laws, we can symbolically execute programs. Below, we give an example
execution of the factorial program from Example~\ref{def:factorial}. We do this by using the simplifier to evaluate the
term $\opsem{id}{pfact(4)}$, where $id$ encodes an arbitrary initial assignment for all the variables.

\begin{example}[Factorial Symbolic Execution] \isalink{https://github.com/isabelle-utp/utp-main/blob/c3e0b2e6fdf2e07985e9b1c887c2fdc44f7f208b/utp/examples/factorial.thy\#L15}
\begin{align*}
  & \opsem{id}{pfact(4)} \\[1ex]
  =~~ & \opsem{id}{x := 4 \relsemi y := 1 \relsemi \while{(x > 1)}{(y := y * x \relsemi x := x - 1)}} \\[1ex]
  =~~& \opsem{[x \smapsto 4]}{y := 1 \relsemi \while{(x > 1)}{(y := y * x \relsemi x := x - 1)}} \\[1ex]
  =~~& \opsem{[x \smapsto 4, y \smapsto 1]}{\while{(x > 1)}{(y := y * x \relsemi x := x - 1)}} \\[1ex]
  =~~& \opsem{[x \smapsto 4, y \smapsto 1]}{(y := y * x \relsemi x := x - 1) \relsemi \while{(x > 1)}{(y := y * x \relsemi x := x - 1)}} \\[1ex]
  =~~& \opsem{[x \smapsto 4, y \smapsto 4]}{x := x - 1 \relsemi \while{(x > 1)}{(y := y * x \relsemi x := x - 1)}} \\[1ex]
  =~~& \opsem{[x \smapsto 3, y \smapsto 4]}{\while{(x > 1)}{(y := y * x \relsemi x := x - 1)}} \\[1ex]
  =~~& \opsem{[x \smapsto 2, y \smapsto 12]}{\while{(x > 1)}{(y := y * x \relsemi x := x - 1)}} \\[1ex]
  =~~& \opsem{[x \smapsto 1, y \smapsto 24]}{\while{(x > 1)}{(y := y * x \relsemi x := x - 1)}} \\[1ex]
  =~~& \opsem{[x \smapsto 1, y \smapsto 24]}{\II\,} & \qed
\end{align*}
\end{example}

\noindent In this case, the program terminates in final state $[x \smapsto 1, y \smapsto 24]$, but in the case of a
non-terminating program the simplifier will loop. Such a symbolic execution engine is a useful tool for simulation of
programs.

Next, we show how we can mechanise the rules of Hoare logic and apply them to verification.

\subsection{Hoare Logic}
\label{sec:hoarelogic}

A Hoare triple $\hoaretriple{b}{P}{c}$ is an assertion that, if program $P$ is started in a state satisfying predicate
$b$, then all final states satisfy predicate $c$. The UTP book~\cite[chapter~2]{Hoare&98} shows how this assertion can be encoded
using a refinement statement. Their definition is mechanised in Isabelle/UTP as given below.

\begin{definition}[Hoare Triple] \label{def:hoaretriple} $\hoaretriple{b}{P}{c} \defs ((\upre{b} \implies \upost{c}) \refinedby P)$ \isalink{https://github.com/isabelle-utp/utp-main/blob/c3e0b2e6fdf2e07985e9b1c887c2fdc44f7f208b/utp/utp_hoare.thy\#L11}
\end{definition}

\noindent This states that the Hoare triple is valid when $P$ is a refinement of the specification
$\upre{b} \implies \upost{c}$. The specification states that when $b$ is satisfied initially, then $c$ is satisfied
finally. This is a statement of partial correctness since $P$ could satisfy this condition by not terminating, for
example if $P = \false$, since $Q \refinedby \false$ for any $Q$. \cref{def:hoaretriple} is equivalent to the standard
definition for partial correctness~\cite{Gordon1989-VerifyHOL} illustrated in \cref{sec:verify-sem}:

\begin{theorem} $\hoaretriple{b}{P}{c} \iff \left(\forall (s, s') @ b(s) \land P(s, s') \implies c(s')\right)$ \isalink{https://github.com/isabelle-utp/utp-main/blob/c3e0b2e6fdf2e07985e9b1c887c2fdc44f7f208b/utp/utp_hoare.thy\#L32}
\end{theorem}

\noindent An equivalent characterisation is given in terms of our reduction relation below.

\begin{theorem} If $\hoaretriple{b}{P}{c}$ then $\forall (\Gamma_1, \Gamma_2) @ [(\Gamma_1 \usubapp b) \land \opsem{\Gamma_1}{P} \rightarrow \opsem{\Gamma_2}{\II\,} \implies (\Gamma_2 \usubapp c)]$ \isalink{https://github.com/isabelle-utp/utp-main/blob/c3e0b2e6fdf2e07985e9b1c887c2fdc44f7f208b/utp/utp_sym_eval.thy\#L55}
\end{theorem}

\noindent This is an example of a UTP linking theorem~\cite{Hoare&98} that relates two semantic presentations, in this
case how the Hoare triple is related to the operational semantics. The satisfaction of $\hoaretriple{b}{P}{c}$ implies
that, for any initial state assignment $\Gamma_1$ satisfying precondition $b$, if $P$ terminates with final state
assignment $\Gamma_2$, then $\Gamma_2$ satisfies $c$. Thus we have formally related the operational and axiomatic
semantics for relational programs. From Definition~\ref{def:hoaretriple} we can also prove, as theorems, the following
Hoare calculus laws:

\begin{theorem}[Hoare Calculus Laws] \label{thm:hoarecalc} $ $ \isalink{https://github.com/isabelle-utp/utp-main/blob/c3e0b2e6fdf2e07985e9b1c887c2fdc44f7f208b/utp/utp_hoare.thy\#L30}
  \vspace{-2ex}

  \begin{center}
  \begin{tabular}{ccccccc}
  \AxiomC{\vphantom{$\hoaretriple{b}{P}{c}$}$[p \implies \sigma \usubapp q]$}
  \UnaryInfC{$\hoaretriple{p}{\uassigns{\sigma}}{q}$}
  \DisplayProof
  &
  \AxiomC{$\hoaretriple{p}{Q_1}{s}$}
  \AxiomC{$\hoaretriple{s}{Q_2}{r}$}
  \BinaryInfC{$\hoaretriple{p}{Q_1 \relsemi Q_2}{r}$}
  \DisplayProof
  &
  \AxiomC{$\hoaretriple{b \land p}{S}{q}$}
  \AxiomC{$\hoaretriple{\neg b \land p}{T}{q}$}
  \BinaryInfC{$\hoaretriple{p}{\conditional{S}{b}{T}}{q}$}
  \DisplayProof
  &
  \AxiomC{$\hoaretriple{p \land b}{S}{p}$}
  \UnaryInfC{$\hoaretriple{p}{\while{b}{P}}{\neg b \land p}$} 
  \DisplayProof
  \end{tabular}
\end{center}
\end{theorem}

\noindent The majority of these laws are standard~\cite{Hoare69}. However, the formulation of the assignment law is more
general than the standard law, $\hoaretriple{p[e/x]}{x := e}{p}$. It avoids the aliasing problem of classical Hoare
logic since substitution depends on semantic independence of lenses, rather than syntactic inequality of variable
names. We can model aliased names by giving a lens two (meta-logical) names in Isabelle, $x$ and $y$. We can also
consider the situation when $x$ and $y$ are different constructions and yet not independent, such as $hp$ and $hp[0]$,
from \cref{ex:st-hp}. We illustrate this with the following example.

\begin{example}[Aliasing] We consider the triple $\hoaretriple{y = 3}{x := 2}{y = 3}$. If we attempt its proof using the
  rules above, and the substitution laws (\cref{thm:substapp}), the result is the predicate
  $(y = 3) \implies (y[2/x] = 3)$. If $x$ and $y$ are identical ($x = y$) then this reduces to
  $(y = 3) \implies \pfalse$, which is not provable. If $x \neq y$, and yet either $x \lequiv y$, $x \lsubseteq y$, or
  $y \lsubseteq x$, then the conjecture is also not provable; indeed we can use
  \textsf{nitpick}~\cite{Blanchette2010-Nitpick,Blanchette2011} to find a counterexample. If, however, $x \lindep y$,
  then we obtain $(y = 3) \implies (y = 3)$, which is true. \qed
\end{example}

\noindent Similarly, from \cref{def:hoaretriple} we can also represent derived laws, like the forward assignment law, by
making use of unrestriction.
\begin{lemma} $\hoaretriple{p}{x := e}{x = e \land p} \quad \text{if } x \unrest e, x \unrest p$ \isalink{https://github.com/isabelle-utp/utp-main/blob/c3e0b2e6fdf2e07985e9b1c887c2fdc44f7f208b/utp/utp_hoare.thy\#L112}
\end{lemma} 

\begin{figure}
  \begin{center}
    \includegraphics[width=.8\linewidth]{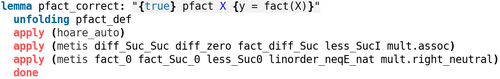}
  \end{center}

  \vspace{-4ex}

  \caption{Verification of Factorial Algorithm}
  \label{fig:verify-fact}
  \vspace{-2ex}
\end{figure}

\noindent We can use these laws to automatically verify our factorial program, following the pattern of previous
works~\cite{Alkassar2008,Nipkow2014-ConcreteSemantics,Armstrong2015}, though using our lens-based Hoare logic laws. We
require, as usual, that the program is annotated with loop invariants. This is shown in Figure~\ref{fig:pfact}, with
loop invariant $y * x! = X!$, which is syntactic sugar to help the proof strategy~\cite{Armstrong2015}. The proof
strategy, implemented in the Isabelle/UTP tactic
\textsf{hoare-auto}, executes the following steps:

\begin{enumerate}
  \item Combine all composed assignments using \cref{thm:asnlaws}.
  \item Apply the Hoare logic laws of \cref{thm:hoarecalc} as deduction rules.
  \item Perform any substitutions in the resulting predicates using \cref{thm:subst} and \cref{thm:substapp}.
  \item Apply the \textsf{rel-auto} tactic of \S\ref{sec:utp-expr} to each resulting predicate.
\end{enumerate}

\noindent The result is a set of HOL predicates that characterise the verification conditions for the program. We exemplify this
with the factorial program below.

\begin{example} $\hoaretriple{\ptrue}{pfact(\mv{X})}{y = X!}$ \isalink{https://github.com/isabelle-utp/utp-main/blob/c3e0b2e6fdf2e07985e9b1c887c2fdc44f7f208b/utp/examples/factorial.thy\#L19}
\end{example}

\begin{proof} Application of \textsf{hoare-auto} yields the following verification conditions:
  \begin{enumerate}
    \item $\mv{x} > 1 \land \mv{y} * \mv{x}! = X! \implies \mv{y} * \mv{x} * (\mv{x} - 1)! = \mv{X}!$
    \item $\mv{x} \le 1 \land \mv{y} * \mv{x}! = X! \implies \mv{y} = \mv{X}!$
  \end{enumerate}
  In this case, both proof obligations can be discharged using \textsf{sledgehammer}~\cite{Blanchette2011}. This is illustrated in
  Figure~\ref{fig:verify-fact}, with calls to \textsf{metis} for each proof obligation.
\end{proof}

\noindent As a further example we consider a program for the find and replace functionality in a text editor.

\begin{example}[Find and Replace] \isalink{https://github.com/isabelle-utp/utp-main/blob/c3e0b2e6fdf2e07985e9b1c887c2fdc44f7f208b/utp/examples/find_replace.thy\#L11}
  $$\begin{array}{l}
  \ckey{function}~\textit{find-replace}(arr : [string]array, len : nat, mtc : string, rpl: string): nat \defs \\
  \quad \ckey{var}~i, occ : nat @ \\
  \quad i := 0 \relsemi occ := 0 \relsemi \\    
  \quad \ckey{while}~i < len ~\ckey{do} \\
  \qquad \ckey{if}~ arr[i] = mtc ~\ckey{then}~ arr[i] := rpl \relsemi occ := occ + 1 ~\ckey{else}~ arr[i] := arr[i] \ckey{fi} \relsemi \\
  \qquad i := i + 1 \\
  \quad \ckey{od} \relsemi \\
  \quad \ckey{return}~occ
  \end{array}
  $$
\end{example}
\noindent In this program, we have an array of strings that represents a sequence of words in a text buffer, and we wish
to replace every occurence of a particular string, $mtc$, with another string $rpl$. The specified function takes four
parameters: an array of strings, $arr$, the length of the array $len$, and the two strings $mtc$ and $rpl$. It returns
the number of occurences for which it has made a replacement.

We first create local variables $i$ and $occ$, using the $\ckey{var}$ construct, and initialise them, for
iterating over the list and recording occurrences. We then have a while loop that iterates over the length of the
array. At each step, if the $i$th element of the array equals $mtc$, then this element is replaced with $rpl$, and the
number of occurences is incremented. Otherwise, the element of the array is left unchanged. Finally, $i$ is incremented
and the loop continues. Once the loop exits, the value of $occ$ is returned. Variables $occ$ and $i$ are modelled in
Isabelle/UTP using record lenses, and the array $arr$ using a total function lens.

The algorithm is mechanised in Isabelle/UTP as shown in Figure~\ref{fig:find-repl}. We first create the state space $ss$
with variables $arr$, $occ$, and $i$. We represent the variables $len$, $mtc$, and $rpl$ as logical variables, since
they are not changed by the algorithm. Moreover, we include an explicit parameter $A$ that represents the initial state
of the array. This is need to represent the invariant annotations, which we explain below.

There are at least two desirable properties for this program, both of which are postconditions: (1) $arr$ should
have all occurrences of $mtc$ replaced with $rpl$; and (2) the value of $occ$ returned should be number of occurrences
in $arr$ initially. The first can be represented by the following triple:
$$
  \hoaretriple{arr = \mv{A}}{\textit{find-replace}(arr, len, mtc, rpl)}{\forall i < len @ \mv{A}[i] = (\conditional{mtc}{arr[i] = rpl}{\mv{A}[i]})}
$$
This states that if the array is initially $\mv{A}$, then following execution of \textit{find-replace}, for every index
$i$ in $\mv{A}$ that has $mtc$, the array contains $rpl$, and all other elements remain the same. This can be proved
using the invariants shown in Figure~\ref{fig:find-repl}. It has two parts corresponding to properties (1) and (2). The
first invariant states that at the $i$th iteration, $arr$ is identical to $A$, except that every instance of $mtc$ up to
index $i$ is replaced with $rpl$. The second invariant states that the number of occurrences is equal to the number of
instances of $mtc$ filtered out of the function before index $i$. The third invariant is that $i \le len$. With these
three invariants, we can verify the two properties; this is illustrated in Figure~\ref{fig:find-repl-correct}.

Thus we have shown how Isabelle/UTP can be applied to practical verification of relational programs, which can contain
complex state space structures afforded by lenses. We can unify different variants of assignment, such as to multiple
variables and collections, using the generalised assignment law. Moreover, the verification infrastructure is semantic
and thus efficient in nature, without the need for deep syntax, and yet can utilise meta-logical concepts like
substitution and unrestriction, which are not present in other shallow embeddings.

\begin{figure}
  \begin{center}
  \includegraphics[width=.62\linewidth]{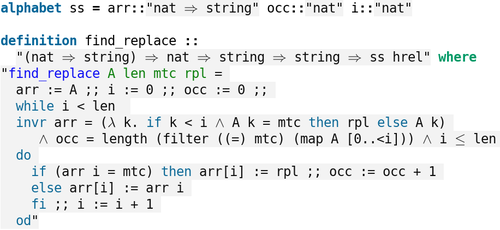}
  \end{center}

  \vspace{-3.5ex}
  
  \caption{Find and Replace in Isabelle/UTP}
  \label{fig:find-repl}
\end{figure}

\begin{figure}
  \begin{center}
  \includegraphics[width=.85\linewidth]{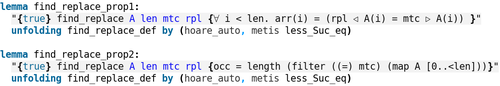}
  \end{center}

  \vspace{-3.5ex}
  
  \caption{Find and Replace verified in Isabelle/UTP}
  \label{fig:find-repl-correct}

  \vspace{-2ex}

\end{figure}

In the next section we consider the mechanisation of Morgan's refinement calculus.

\subsection{Refinement Calculus}
\label{sec:ref-calc}

In this section, we show how lenses can support a Morgan-style refinement calculus~\cite{Mor1996}, extending our
previous results on frames in UTP~\cite{Foster18b}. The refinement calculus is a technique for stepwise development of
programs, by refining abstract specifications in increasingly concrete programs. The specification statement,
$w{:}[pre, post]$, is a contract asserting that (1) only the variables in $w$ may change, and (2) the postcondition
($post$) should be established when the precondition ($pre$) holds. We use lenses and our relational program model to
mechanise a partial correctness specification statement, and prove several laws relating to frames. Ordinarily, Morgan's
calculus specified total correctness, but here we focus on the laws relating to frames, rather than termination of
programs. A total correctness calculus can be obtained straightforwardly by combining our results here and in
\cref{sec:utp-thy} with the UTP theory of designs~\cite{Hoare&98,Cavalcanti04}.

For this, we first define a notion of equivalence modulo a lens that is needed to characterise frames.

\begin{definition}[Observation Equivalence] We assume $s_1, s_2 : \src$ and $X : \view \lto \src$, and define: \isalink{https://github.com/isabelle-utp/utp-main/blob/0189a374bad44431e6f5a89b50cc683ec6ef5221/optics/Lens_Laws.thy\#L61}
  $$s_1 \simeq_X s_2 \defs (s_1 = \lput_X~s_2~(\lget_X~s_1))$$
\end{definition}

\noindent This operator states that two states $s_1$ and $s_2$ are the same everywhere outside of the region described
by $X$. For example, if we have three variables $x, y, z : \num$, then we can show that
$$[x \mapsto 1, y \mapsto 2, z \mapsto 3] \simeq_x [x \mapsto 2, y \mapsto 2, z \mapsto 3]$$ since the two states differ
only in the value of $x$. Lenses can represent sets of variables, using $+$, and so we can group several independent
variables in a frame, as the following equivalences illustrate:
\begin{lemma} Assume $s_1, s_2 : \src$ then the following identities hold: \isalink{https://github.com/isabelle-utp/utp-main/blob/0189a374bad44431e6f5a89b50cc683ec6ef5221/optics/Lens_Algebra.thy\#L415}
$$(s_1 \simeq_{\lone_\src} s_2) = \ptrue \text{~and~} (s_1 \simeq_{\lzero_\src} s_2) = (s_1 = s_2)$$  \end{lemma}
\noindent This theorem demonstrates the intuition that $\lone$ and $\lzero$ characterise the entirety and none of the state space, respectively.
The state space outside of $\lone$ is empty, and outside of $\lzero$ is the whole state. We can also show that $\simeq_X$ is an equivalence relation.
\begin{lemma} If $X$ is a total lens then $\simeq_X$ is an equivalence relation on states. \isalink{https://github.com/isabelle-utp/utp-main/blob/0189a374bad44431e6f5a89b50cc683ec6ef5221/optics/Lens_Laws.thy\#L207}
\end{lemma}
\noindent We now use observation equivalence to define the specification statement.

\begin{definition}[Specification Statement] \label{def:morgan} \isalink{https://github.com/isabelle-utp/utp-main/blob/0189a374bad44431e6f5a89b50cc683ec6ef5221/utp/examples/morgan_rcalc.thy\#L9}
  $$w{:}[pre, post] \defs ((\upre{pre} \implies \upost{post}) \land \upre{\uv} \simeq_w \upost{\uv})$$
\end{definition}

\noindent A specification statement states that the precondition implies the postcondition, and in addition the initial
state is the same as the final state, modulo $w$. The semantics is similar to the partial correctness Hoare triple in
Definition~\ref{def:hoaretriple}, except that it only includes the specification. Indeed, we can prove the following
well-known correspondence~\cite{Hoare&98,Gomes2016}.

\begin{theorem} $\hoaretriple{pre}{P}{post} \iff (\uv{:}[pre, post] \refinedby P)$ \isalink{https://github.com/isabelle-utp/utp-main/blob/2376b415962a921172bb5586bf2aa1688d740dd9/utp/examples/morgan_rcalc.thy\#L59}
\end{theorem}

\noindent A Hoare triple can be re-expressed as a specification statement refinement. The standard triple does not
include a frame, and consequently all variables are permitted to change, indicated by the frame being $\uv$
(i.e. $\lone$). Moreover, from Definition~\ref{def:morgan} we can derive many of the refinement calculus laws \cite{Mor1996} as
theorems.

\begin{theorem}[Refinement Calculus]
  If $x$ and $w$ are total lenses, then the following laws hold: \isalink{https://github.com/isabelle-utp/utp-main/blob/2376b415962a921172bb5586bf2aa1688d740dd9/utp/examples/morgan_rcalc.thy\#L14}
  \begin{align}
    w{:}[pre, post]    &\refinedby w{:}[pre, post'] & \text{if } post' \implies post \label{law:RC1} \\
    w{:}[pre, post]    &\refinedby w{:}[pre', post] & \text{if } pre \implies pre' \label{law:RC2} \\
    x,w{:}[pre, post]  &\refinedby w{:}[pre, post] & \text{if } x \lindep w \label{law:RC3} \\
    w{:}[pre, post]    &\refinedby \II & \text{if } pre \implies post \label{law:RC4} \\
    w{:}[pre, post]    &\refinedby x := e & \text{if } x \lsubseteq w, pre \implies post[e/x] \label{law:RC5} \\
    (w := e)           &= w{:}[\ptrue, w = e] & \text{if } w \unrest e \label{law:RC6} \\
    w{:}[pre, post]    &\refinedby w{:}[pre, mid] \relsemi w{:}[mid, post] \label{law:RC7} \\
    w{:}[pre, post]    &\refinedby \conditional{w{:}[b \land pre, post]}{b}{w{:}[(\neg b) \land pre, post]} \label{law:RC8}
  \end{align}
\end{theorem}

\noindent These laws use the lens operators as side conditions. Laws~\eqref{law:RC1} and \eqref{law:RC2} allows
strengthening and weakening of a post- and precondition, respectively. Law~\eqref{law:RC3} shows that a variable can be
removed from a frame in a refinement, because this effectively adds $\upost{x} = \upre{x}$ as a
conjunct. Laws~\eqref{law:RC4} and \eqref{law:RC5} show the circumstances under which a specification may be refined to
a skip ($\II$) or assignment ($x := e$). For the latter, $x$ must be in the frame $w$, represented by $x \lsubseteq
w$. Law~\eqref{law:RC6} shows that an assignment $w := e$ can be written as a specification statement when $e$ does not
depend on $w$. Law~\eqref{law:RC7} shows how a specification can be divided into two sequential specifications by inserting
a midpoint condition. Finally, \eqref{law:RC8} shows how a conditional can be introduced, by conjoining the precondition
with predicate $b$.

These results illustrate how lenses support modelling of frames. The theorems can be applied to derivation of algorithms
by stepwise refinement in Isabelle/UTP. We have therefore shown in this section how our lens-based semantic foundation
allows us to support a selection of verification calculi and linking theorems between them. In the next section we
consider mechanisation of UTP theories.

\section{Mechanising UTP Theories}
\label{sec:utp-thy}
In this section we apply Isabelle/UTP to the mechanisation of UTP theories. The relational program model in
\S\ref{sec:isabelle-utp} is powerful, but not without limitations. It is well known, for instance, that simple
relational programs, which capture only the possible initial and final values of program variables, cannot adequately
differentiate terminating and non-terminating behaviour~\cite{Cavalcanti&06}. Furthermore, several programming
paradigms, such as real-time, concurrency, and object orientation, require a richer semantic model with more observable
information~\cite{Cavalcanti&06,SCS06,Sherif2010,Foster19c-HybridRelations}.

This semantic enrichment can be facilitated by UTP theories. Semantic information is expressed by adding special
observational variables, which encode quantities of a program or model, and invariants that restrict their domain,
called healthiness conditions. The observational variables can be used to define specialised operators for a particular
computational paradigm, such as a delay or deadline operator for a real-time language. For example, a clock variable,
$clock : \nat$ can be added to record to passage of time, or a trace variable $tr : [\textit{Event}]list$ can be added
to record events~\cite{Cavalcanti&06}. The healthiness conditions allow us to impose well-formedness invariants over
these observational variables, and allow us to prove algebraic theorems that characterise the healthy elements. As
usual, healthiness conditions are represented as idempotent predicate transformers, that is, functions on relations over
the observational variables. The pay-off here is that general properties of elements of the UTP theory can often be
reduced to properties of the healthiness conditions~\cite{Hoare&98,Foster17c}.

We characterise UTP theories in Isabelle/UTP as follows.

\begin{definition} \label{def:isa-utp-theory} An Isabelle/UTP theory is a pair ($\src$,
  $\mathcal{H}$), where $\src$ is an observation (i.e. state) space, $\mathcal{H} : [\src]\hrel \to [\src]\hrel$ is a healthiness
  condition, and $\mathcal{H}$ is idempotent: $\mathcal{H} \circ \mathcal{H} =
  \mathcal{H}$. \isalink{https://github.com/isabelle-utp/utp-main/blob/2376b415962a921172bb5586bf2aa1688d740dd9/utp/utp_theory.thy\#L109}
\end{definition}

\noindent We mechanise this algebraic structure in Isabelle using locales~\cite{Ballarin06}. $\src$ characterises the
observations that can be made of the model. An observation space can be constructed using the \ckey{alphabet} command
(\S\ref{sec:mechss}), in which case it defines a set of lenses, $x_1 \cdots x_n$, which provide the observational
variables. The healthiness conditions are encoded as total endofunctions on homogeneous relations parameterised by
$\src$. If a theory has multiple healthiness conditions, then these can be composed function-wise. The simplest UTP
theory is the relational theory, $Rel_\alpha \defs (\alpha, \lambda X @ X)$. Any type is an instance of $\alpha$, and
any relation is a fixed-point of $\lambda X @ X$.

We say that a relation $P : [\src]\hrel$ is $\mathcal{H}$-healthy, written $P \is \mathcal{H}$, when $P$ is a
fixed-point of $\mathcal{H}$: $\mathcal{H}(P) = P$. Moreover, we characterise the healthy elements of a theory by the set
$\theoryset{\mathcal{H}} \defs \{P | P \is \mathcal{H}\}$. Idempotence of $\mathcal{H}$ guarantees the existence of at
least one fixed point. This ensures that the UTP theory is non-empty since, for any relation $P$, $\mathcal{H}(P) \in
\theoryset{\mathcal{H}}$ since $\mathcal{H}(\mathcal{H}(P)) = \mathcal{H}(P)$. To illustrate, we formalise the UTP
theory for timed relational programs introduced in \S\ref{sec:UTP}.

\begin{example}[Timed Relations] \label{ex:timerel} $ $ \isalink{https://github.com/isabelle-utp/utp-main/blob/master/tutorial/utp_simple_time.thy}

  \noindent The parametric observation space $\ssdef{[S]\ukey{rt}}{clock : \mathbb{N}, st : S}$ has lenses $clock : \mathbb{N}$,
  which denotes the passage of time, and $st : S$, which denotes the user state. Healthiness function
  $$\healthy{HT}(P) \defs (\upre{clock} \le \upost{clock} \land P)$$ ensures that time advances. $\healthy{HT}$ is
  clearly idempotent, since conjunction is idempotent. We define the following operator for introducing a delay for
  timed relations:
  \begin{align*}
    \ckey{wait} &: \mathbb{N} \to [[S]\ukey{rt}]\hrel \\
    \ckey{wait}(n) & \defs (\upost{clock} = \upre{clock} + n \land \upost{st} = \upre{st})
  \end{align*}
  Here, $clock$ is advanced by $n$, and the state is unchanged. We can prove that $\ckey{wait}(n)$ is
  $\healthy{HT}$-healthy, since it advances time. Conversely, the predicate $\upost{clock} = \upre{clock} - 1$ is not
  healthy, since it tries to reverse time. The mechanisation of the theory in Isabelle/UTP is shown in
  Figure~\ref{fig:time-rel}. \qed
\end{example}

As in previous work~\cite{Feliachi2010,Zeyda2012}, our theories form ``families'': they characterise relations on
several observation spaces, since the observation type is potentially polymorphic, and thus extensible with additional
variables. For example, our theory in Example~\ref{ex:timerel} is parametric in $S$. We can characterise a subtheory
relation between UTP theories.

\begin{definition} \label{def:subtheory} $([\vec{\beta}]\mathcal{T}_2, \mathcal{H}_2)$ is a subtheory of $(\src_1, \mathcal{H}_1)$ when $[\vec{\beta}]\mathcal{T}_2$
  specialises $\src_1$, and $\theoryset{\mathcal{H}_2} \subseteq \theoryset{\mathcal{H}_1}$.
\end{definition}

\noindent Subtheories allow us to arrange theories in a hierarchy with descendants that specialise the observation space
with additional variables and constraints.

UTP theories also include a signature: a set of operators that construct healthy elements of the theory. We say that an
operator
$$F : [\src]\hrel^{\,n} \to [\src]\hrel$$ in $n$ parameters, is
in the signature if $\theoryset{\mathcal{H}}$ is closed under $F$. Formally, we require that:
$$\forall P_1, \cdots, P_n ~ @ ~ P_1 \is \mathcal{H} \land \cdots \land P_n \is \mathcal{H} ~~ \Longrightarrow ~~ F(P_1, \cdots,
P_n) \is \mathcal{H}$$

\noindent The function $F$ can either denote a new operator defined on $\src$, or an existing operator defined over a
parent observation space. For example, nondeterministic choice $P \sqcap Q$ and sequential composition $P \relsemi Q$
often inhabit the signature of several theories, since they are typed by arbitrary relations and so can be instantiated
by any observation space. This also means that the corresponding algebraic laws for a parent operator $F$ can be
directly applied to elements of a new subtheory. For example, sequential composition is always associative, since every
theory is a subtheory of $Rel_\alpha$.

We exemplify this by giving the signature of our theory of timed relations.

\begin{lemma} $\theoryset{\healthy{HT}}$ is closed under the following relational operators: $\II$, $\relsemi$,
  $\conditional{}{b}{}$, and $x := v$ when $x \lindep clock$, and so these are also within the theory signature, as
  demonstrated by the laws below: \isalink{https://github.com/isabelle-utp/utp-main/blob/eead352783b456b61220e7a2157cef5dce4abfbc/tutorial/utp_simple_time.thy\#L80}

  \begin{center}
  \begin{tabular}{cccccc}
  \AxiomC{---\vphantom{$P$}}
  \UnaryInfC{$\II \is \healthy{HT}$}
  \DisplayProof
   &
  \AxiomC{$P \is \healthy{HT}$}
  \AxiomC{$Q \is \healthy{HT}$}
  \BinaryInfC{$(P \relsemi Q) \is \healthy{HT}$}
  \DisplayProof
   &
  \AxiomC{$P \is \healthy{HT}$}
  \AxiomC{$Q \is \healthy{HT}$}
  \BinaryInfC{$(\conditional{P}{b}{Q}) \is \healthy{HT}$}
  \DisplayProof
  &
  \AxiomC{$x \lindep clock$\vphantom{$P$}}
  \UnaryInfC{$(x := v) \is \healthy{HT}$}
  \DisplayProof
  &
  \AxiomC{---\vphantom{$P$}}
  \UnaryInfC{$\ckey{wait}(n) \is \healthy{HT}$}
  \DisplayProof  
  \end{tabular}
\end{center}
\end{lemma}

\noindent Once the signature operators have been established, the final step is to prove the characteristic algebraic
theorems for them. These laws effectively provide an algebraic semantics for the UTP theory, and can be used to aid in
the construction of program verification tools. There are three ways to obtain such theorems in Isabelle/UTP. Firstly,
we can inherit them from a parent theory by utilising Definition~\ref{def:subtheory}. Secondly, we can import them from
an algebraic structure by proving the algebra's axioms. Thirdly, we can prove them manually using Isabelle/UTP's proof
tactics.

\begin{figure}
  \begin{center}
    \includegraphics[width=.95\linewidth]{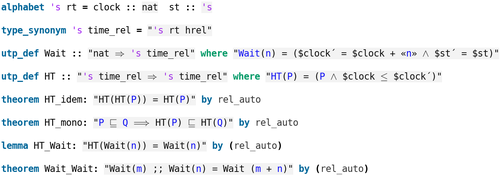}
  \end{center}

  \vspace{-2ex}
  
  \caption{Timed Relations in Isabelle/UTP}
  \label{fig:time-rel}

    \vspace{-2ex}
\end{figure}

We exemplify the third approach for the new delay operator, using our \textsf{rel-auto} proof tactic
(\S\ref{sec:utp-expr}).

\begin{theorem}[Delay Laws] \isalink{https://github.com/isabelle-utp/utp-main/blob/12a647e27a10f89cd527a89918591c6390a642bd/tutorial/utp_simple_time.thy\#L111}
  \begin{align*}
    \ckey{wait}(0) &~=~ \II \\
    \ckey{wait}(m) \relsemi \ckey{wait}(n) &~=~ \ckey{wait}(m + n) \\
    \ckey{wait}(m) \relsemi (\conditional{P}{b}{Q}) &~=~ \conditional{(\ckey{wait}(m) \relsemi P)}{b}{(\ckey{wait}(m) \relsemi Q)}  & \text{if } clock \unrest b \\
    \ckey{wait}(m) \relsemi x := v &~=~ x := v \relsemi \ckey{wait}(m) & \text{if } clock \unrest v, x \lindep clock
  \end{align*}
\end{theorem}

\noindent Waiting for a zero length duration is simply a skip operation ($\II$), sequential composition of two delays,
by $m$ and $n$ time units, is equivalent to a single $m + n$ delay. A delay distributes through a conditional provided
that $b$ does not refer to the $clock$ variable. A delay commutes with an assignment ($x := v$) provided that the delay
expression does not depend on $x$ and $v$ does not depend on $clock$.

Algebraic theorems for UTP theories can also be obtained by linking to a variety of mechanised algebraic structures. The
\textsf{HOL-Algebra} library~\cite{HOL-Algebra}, for example, characterises the axioms of partial orders, lattices,
complete lattices, Galois connections, and the myriad of theorems that can be derived from them. A substantial advantage
of the UTP approach is that algebraic theorems can often be reduced to proving properties of the underlying healthiness
conditions, which allows us to obtain laws with minimal effort. We exemplify this by restricting ourselves to a
particular subclass of continuous UTP theories.

\begin{definition} \label{def:cont-utp-thy} $\mathcal{H}$ is continuous provided that
  $\mathcal{H}(\bigsqcap_{i \in A} P(i)) = \bigsqcap_{i \in A} \mathcal{H}(P(i))$ whenever $A \neq \emptyset$. \isalink{https://github.com/isabelle-utp/utp-main/blob/12a647e27a10f89cd527a89918591c6390a642bd/utp/utp_healthy.thy\#L112}
\end{definition}

\noindent In a continuous theory, the healthiness condition distributes through arbitrary non-empty relational infima. A
corollary of this definition is that $\mathcal{H}$ is also monotonic:
$P \refinedby Q \implies \mathcal{H}(P) \refinedby \mathcal{H}(Q)$. So, by the Knaster-Tarski theorem, also part of
\textsf{HOL-Algebra}, we can show that $\theoryset{\mathcal{H}}$ forms a complete lattice under refinement
$\refinedby$. This allows us to import theorems for recursive and iterative programs into a UTP theory.

From the induced complete lattice, we obtain the following operators: infimum
$\textstyle{\bigsqcap_{\thy{\mathcal{T}}}}$, supremum $\textstyle{\bigsqcup_{\thy{\mathcal{T}}}}$, top element
$\top_{\thy{\mathcal{T}}}$, bottom element $\bot_{\thy{\mathcal{T}}}$, and least fixed-point $\mu_{\thy{\mathcal{T}}}$,
which are all in the theory's signature, and for which the usual complete lattice and fixed-point calculus
theorems~\cite{Backhouse1995} hold. Moreover, in a continuous UTP theory, these operators can be calculated using the
equational theorems given below.

\begin{lemma}[Continuous Theory Properties] \label{thm:contthy} \isalink{https://github.com/isabelle-utp/utp-main/blob/12a647e27a10f89cd527a89918591c6390a642bd/utp/utp_theory.thy\#L284}
  \begin{align*}
    \top_{\thy{\mathcal{T}}} &= \mathcal{H}(\false) \\ 
    \bot_{\thy{\mathcal{T}}} &= \mathcal{H}(\true) \\ 
    \textstyle{\bigsqcap_{\thy{\mathcal{T}}}} \, A &= \left(\conditional{\top_{\thy{\mathcal{T}}}}{A = \emptyset}{\textstyle{\bigsqcap A}}\right) & \text{if } A \subseteq \theoryset{\mathcal{H}} \\
    \mu_{\thy{\mathcal{T}}} \, F &= \mu X @ F(\mathcal{H}(X)) & \text{if } F : \theoryset{\mathcal{H}} \to \theoryset{\mathcal{H}} \textnormal{ and } F \textnormal{ is monotonic}
  \end{align*}
  
\end{lemma}

\noindent These lemmas demonstrate the relationship between the theory operators, and the relational ones defined in
\S\ref{sec:isabelle-utp}. The theory top and bottom elements are the relational top and bottom with $\mathcal{H}$
applied. The theory infimum $\textstyle{\bigsqcap_{\thy{\mathcal{T}}}} \, A$ is $\top_{\thy{\mathcal{T}}}$ when $A$ is
empty, and otherwise the relational infimum. Moreover, the least fixed-point operator can be expressed as the relational
least fixed-point by precomposition of $\mathcal{H}$. The utility of these theorems is that we can construct
nondeterministic choices and recursive programs using the relational operators $\sqcap$ and $\mu$, which again allows
reuse of their algebraic laws.

We now demonstrate how we can obtain such theorems for timed relations.

\begin{theorem}
  \noindent $\healthy{HT}$ is continuous, since conjunction distributes through disjunction ($\bigsqcap$). We therefore
  obtain a complete lattice, and can calculate the top and bottom element: \isalink{https://github.com/isabelle-utp/utp-main/blob/12a647e27a10f89cd527a89918591c6390a642bd/tutorial/utp_simple_time.thy\#L50}
  $$\top_{\!{\ukey{rt}}} = \healthy{HT}(\false) = \false \qquad \bot_{{\ukey{rt}}} = \healthy{HT}(\true) = (clock \le
  clock')$$
  We also obtain a least fixed-point operator for constructing and reasoning about iteration, $\mu_{\ukey{rt}}$, and can
  use it to denote a timed loop operator: $b \oast_{\ukey{rt}} P \defs (\mu_{\ukey{rt}} X @ \conditional{P \relsemi X}{b}{\II})$.
\end{theorem}

\noindent From these theorems, we can proceed to define a timed Hoare calculus and proof tactics following the
template given in \S\ref{sec:verify-rel}. A concrete timed program can be modelled by creating a suitable state
space, $\ssdef{vars}{x_1 : \tau_1 \cdots x_n : \tau_n}$, defining the program $P : [[vars]\ukey{rt}]\hrel$ using the
signature operators, and proving that $P \is \healthy{HT}$. Finally, the proven Hoare logic deduction rules can be
used for program verification.

In this section, we have demonstrated how a UTP theory can be mechanically validated, a set of signature operators
verified, and algebraic theorems for these operators proved. Though the example of timed relations is simple, it
illustrates the basic concepts that we have used for more complex theories in Isabelle/UTP. Notably, we have applied our
approach to the mechanisation of a hierarchy of UTP theories for reactive programs that uses five observational
variables and seven healthiness conditions~\cite{Foster17c}. This UTP theory has been applied to the development of a
verification tool for the \Circus language~\cite{Foster18a}, which extends our imperative programming language with
concurrency and communication~\cite{Oliveira&09}.

\section{Related work}
\label{sec:relwork}
There have been several previous works on application of algebraic semantics to modelling mutation of state. Back and von
Wright~\cite{Back1998} use two functions $\textsf{val}.x : \src \to \view$ and
$\textsf{set}.x : \view \to \src \to \src$, to characterise each variable, together with five axioms that characterise
their behaviour. 

\begin{definition}[Back and von Wright's Variable Axioms~\cite{Back1998,Back2005-Procedures}] \label{def:back-var}
\begin{align}
  \textsf{val}.x.(\textsf{set}.x.a.\sigma) &= a \label{law:bv-put-get} \\
  x \neq y ~\implies~ \textsf{val}.y.(\textsf{set}.x.a.\sigma) &= \textsf{val}.y.\sigma \label{law:bv-put-irr} \\
  \textsf{set}.x.a \relsemi \textsf{set}.x.b &= \textsf{set}.x.b \label{law:bv-put-put} \\
  x \neq y ~\implies~ \textsf{set}.x.a \relsemi \textsf{set}.y.b &= \textsf{set}.y.b \relsemi \textsf{set}.x.a \label{law:bv-put-comm} \\
  \textsf{set}.x.(\textsf{val}.x.\sigma).\sigma &= \sigma \label{law:bv-get-put}
\end{align}
\end{definition}

\noindent The dot operator ($f.x$) denotes function application in Back and von Wright's work. From these axioms, and a
predicate transformer semantics, they are able to prove the laws of programming~\cite{Back1998}, including
\cref{ex:asn-comm}. The lens functions $\lput$ and $\lget$ correspond to the variable functions \textsf{set} and
\textsf{val}, and have almost the same axioms~\cite{Fischer2015}. Specifically, the axioms of total lenses in
\cref{def:total-lenses} correspond to Back and von Wright's axioms \eqref{law:bv-put-get}, \eqref{law:bv-put-put}, and
\eqref{law:bv-get-put} in Definition~\ref{def:back-var}. The other two axioms, \eqref{law:bv-put-comm} and
\eqref{law:bv-put-irr}, are captured by lens independence (\cref{def:lens-indep}). In fact, \cref{thm:indep-total}
demonstrates that only four of Back and von Wright's variable axioms are necessary~\cite{Back1998};
axiom~\eqref{law:bv-put-irr} can be proved from the others and so is redundant.

We have based our work on lenses rather than Back and von Wright's approach, which is also the reason for the different
order in the parameters of $\lput$~\cite{Foster07,Foster09}. Nevertheless, we effectively use lenses to relax Back and
von Wright's axioms so that we can characterise, not only independence of variables, but also when one variable is
``part of'' another. Specifically, in Definition~\ref{def:back-var}, $x$ and $y$ are individual variables, but with
lenses they can equally refer to a set or hierarchical structure. Variable sets are also supported by Back and
Preoteasa~\cite{Back2005-Procedures} by lifting these axioms over a list, for instance to support multiple assignment.
Our approach has more generality since it allows hierarchy and does not require formalising variable names. Also similar
to our work, Back and Preoteasa~\cite{Back2005-Procedures} use their \textsf{set} variable operator to characterise
substitution.

J.~Foster et al.~\cite{Foster09} created lenses as an abstraction for bidirectional programming and solving the
view-update problem in database theory~\cite{Bancilhon1981-ViewUpdate}, which seeks to propagate changes to a computed
view table back to the original source table. Fischer et al. \cite{Fischer2015} give a detailed study of the algebraic
laws for $\lget$ and $\lput$, which is the starting point for our work. Foster et al. provide combinators for composing
abstract lenses and concrete lenses based on trees~\cite{Foster07,Foster09}. They have been practically applied in the
\emph{Boomerang} language\footnote{Boomerang home page: \url{http://www.seas.upenn.edu/~harmony/}} for transformations
on textual data structures, such as XML databases. There is also a substantial Haskell library\footnote{Lenses, Folds
  and Traversals. \url{https://hackage.haskell.org/package/lens}} that can be used to develop generic data type
transformations. Pickering et al.~\cite{Pickering2017-Optics} extend this library to support modular data accessors by
introducing various additional combinators, including a form of parallel composition. Our lens sum operator
(\cref{def:lens-sum}) share similarity with the parallel composition operator defined by Pickering et
al.~\cite{Pickering2017-Optics}. It is different, however, because the latter acts on an explicitly separated state
space of the form $A \times B$, whereas lens sum acts on an arbitrary state space.

Variables are also given an algebraic semantics by Dongol et al.~\cite{Dongol19}, through the development of Cylindrical
Kleene Algebra, that extends their previous work~\cite{Armstrong2015,Gomes2016}. The core of their approach is a
cylindrification operator, $C_x~R$, that liberates the variable $x$ in relation $R$. Their work contains a comprehensive
investigation of the algebraic properties of this operator in an algebraic program model. They apply this algebra to
characterise assignment, substitution, and frames. They use a functional model of state~\cite{Schirmer2009}, and thus
need to fix types for both variables and values, as explained in \cref{sec:state-space-modelling}. We hope in the future
to apply their algebraic structures and generalise them in the context of lenses, building on our previous
results~\cite{Foster16a,Foster18b}. For comparison, relational cylindrification can effectively be implemented using our
lens summation and quantifier operators as $C_x~R \defs \exists \, (\upre{x} \lplus \upost{x}) @ R$.

\section{Conclusions}
\label{sec:concl}
In this paper we have given a comprehensive exposition of the foundations of our verification framework,
Isabelle/UTP. As we have demonstrated, Isabelle/UTP provides a unified semantic foundation for a variety of
computational paradigms and programming languages, with the ability to formulate machine-checked semantic models. These
models can then be used in constructing verification tools that harness the powerful automated proof facilities of
Isabelle/HOL.

We have described how variables can be modelled as algebraic objects using lenses. This allows us to unify and compose a
variety of variable models, and define generic operators for their comparison and manipulation. The lens-based model, in
particular, allows us to avoid dependence on syntactic concepts like naming, and instead build upon semantic properties
like independence and containment.

We used this algebraic foundation to develop a flexible shallow embedding of UTP, including expressions, predicates, and
relations. We provided proof tactics for relational conjectures, and a library of algebraic theorems. In contrast to
previous work, our expression model additionally supports syntax-like manipulations, including substitution and
unrestriction, without the need for explicit modelling of variable names. A particularly pleasing result is the semantic
unification of assignment and substitution, which leads to several elegant algebraic laws. Our mechanisation therefore
offers both efficient automated proof and expressivity.

We have used this relational model to develop programming operators, and prove the ``laws of
programming''~\cite{Hoare87} as theorems, along with the axioms of several algebraic structures, like complete lattices
and cylindrical algebra. We have then used this body of laws in \S\ref{sec:verify-rel} to provide symbolic execution,
Hoare logic verification, and refinement calculus components for relational programs. We show how lenses allow us to
perform symbolic evaluation, reason about aliasing of different variables, and model frames and Morgan's specification
statement~\cite{Mor1996}. Though we use UTP notation throughout, different syntactic flavours can easily be accommodated
using Isabelle's powerful syntax processing facilities~\cite{Tuong2019-CIsabelle}.

Finally we have described how the relational model can be applied to mechanised UTP theories. We have shown how UTP
theories are represented, in terms of observation spaces and healthiness conditions. Our theory model supports reuse of
algebraic laws by a notion of inheritance, whereby a UTP theory can be extended with additional observational quantities
and refined by further healthiness conditions. We have shown how further laws can be obtained with links to algebraic
structures.

In conclusion, we have made the first steps toward realising the UTP vision~\cite[Chapter~0]{Hoare&98} of integrated
formal methods, underpinned by unifying semantics, with mechanised support. In many ways, this is only the beginning, as
there remains a large number of computational paradigms that are not yet represented in Isabelle/UTP. A major item for
future work is the mechanisation of the various pen-and-paper UTP theories that have been developed by the UTP community
over the past twenty years~\cite{Hoare&98,Cavalcanti04,SCS06,Sherif2010,Bresciani2012}. We are currently working on
various UTP theories to support verification of low-level code, hybrid
programs~\cite{Foster19c-HybridRelations,Foster2020-dL}, control law block diagrams~\cite{Foster16b}, and state
machines~\cite{Foster18b}. Moreover, we also plan to analyse the efficiency of the various proof tactics described here
and see if they scale to larger models.

With respect to verification of low-level code, one of the major features needed is a model of dynamically allocated
memory addressed by pointers. We have hinted back in \S\ref{sec:lensax} that we can weaken the axioms of total lenses to
describe ``partial lenses''. These are only defined for a subset of the possible states, which can be used to
distinguish allocated and dangling pointers. We are currently using this idea to mechanise separation
logic~\cite{Reynolds2002} in Isabelle/UTP, by combining lenses and separation
algebra~\cite{Calcagno2007,Dongol2015,Foster20-LocalVars}. In tandem with this, we plan to study the algebraic structure
of lenses in more depth and from the perspective of category theory, due to heterogeneous nature of the equivalence
operator. In particular, we will explore the links between lenses and recent work on cylindrical Kleene
algebra~\cite{Dongol19}.

With respect to verification of hybrid dynamical systems, we have previously described an extension of the relational
calculus with continuous variables~\cite{Foster16b,Foster17b}. The generic trace model of reactive
designs~\cite{Foster17b} and our extensible mechanisation of UTP theories means that we can specialise the reactive
design hierarchy in a different direction to describe hybrid reactive systems, where the trace is a piecewise continuous
function. This can be applied, for instance, to assign a UTP semantics to a language like Hybrid CSP~\cite{He94}, which
integrates continuous evolution with discrete CSP-style events and concurrency. Proof support for such a language
depends on the ability to reason about invariants of differential equations, and so we are also integrating Platzer's
differential induction technique~\cite{Platzer2008,Foster2020-dL}, as employed by the KeYmaera tool\footnote{KeYmaera:
\url{http://symbolaris.com/info/KeYmaera.html}}, into Isabelle/UTP. Moreover, we have previously given a UTP semantics
to the dynamical systems modelling language Modelica~\cite{Foster16b}, and so we will also develop proof facilities for
this.

\section*{Acknowledgements}

\noindent We would like thank Prof. Burkhart Wolff for his invaluable feedback on our work, and for first pointing us in
the direction of lenses as a possible research direction. We also thank the anonymous reviewers of our paper, whose
thorough and constructive comments on our work have greatly enhanced its quality. This work is funded by the EPSRC
projects CyPhyAssure\footnote{CyPhyAssure Project: \url{https://www.cs.york.ac.uk/circus/CyPhyAssure/}} (Grant
EP/S001190/1), RoboCalc\footnote{RoboCalc Project: \url{https://www.cs.york.ac.uk/circus/RoboCalc/}} (Grant
EP/M025756/1), RoboTest (EP/R025479/1), and the Royal Academy of Engineering.

\section*{References}

\bibliographystyle{plain}
\bibliography{FACS-SCP}

\begin{figure}
  \begin{center}
    \includegraphics[width=.9\linewidth]{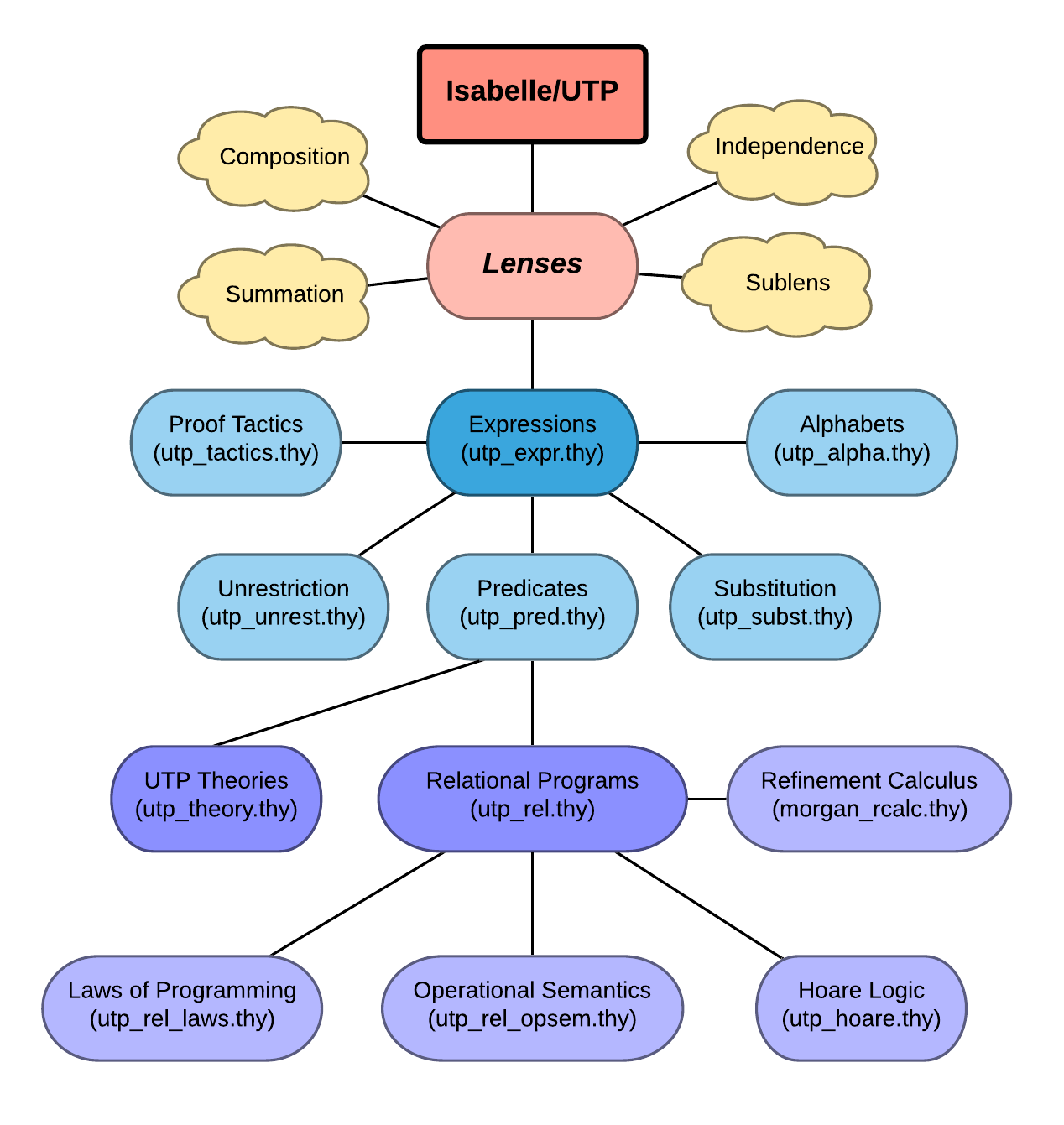}
  \end{center}

  \caption{Main Isabelle/UTP Concepts and Theories}

  \label{fig:UTP-Concepts}
\end{figure}

\end{document}